\documentclass[12pt,a4paper]{article}
\usepackage[utf8]{inputenc}
\usepackage{amsmath}
\usepackage{amsfonts}
\usepackage{amssymb}
\usepackage{gensymb}
\usepackage{amsthm}
\usepackage{graphicx}
\author{Kostyantyn Mazur*}
\title{A Partial Solution to Continuous Blotto}
\begin{document}
\maketitle

\begin{center}
\textbf{Abstract}
\end{center}

This paper analyzes the structure of mixed-strategy equilibria for Colonel Blotto games, where the outcome on each battlefield is a polynomial function of the difference between the two players' allocations. This paper severely reduces the set of strategies that needs to be searched to find a Nash equilibrium. It finds that there exists a Nash equilibrium where both players' mixed strategies are discrete distributions, and it places an upper bound on the number of points in the supports of these discrete distributions.

\vspace{2 in}

* Kostyantyn Mazur, Tandon School of Engineering, New York University, 6 Metrotech Center, Brooklyn, NY; email: km1883@nyu.edu

I would like to acknowledge Dr. Laurent Mathevet (Department of Economics, New York University) and Dr. Edward Miller (Department of Mathematics, Tandon School of Engineering, New York University) for their inputs for this paper.

\newpage

\newtheorem{lem}{Lemma}

\theoremstyle{remark}
\newtheorem{ex}{Example}
\theoremstyle{remark}
\newtheorem{calc}{Calculation}

\tableofcontents

\newpage

\section{Introduction}
Colonel Blotto is a two-player game, whose prototypical version has two colonels fighting a battle with each other on multiple fronts. In an economics context, the colonels might be, for instance, two firms competing simultaneously in multiple geographical areas, or two political candidates seeking to ``win'' as many places as possible. Each colonel has only a fixed number of troops (or resources) to distribute among the fronts, and whichever colonel assigns more troops on a front wins that front. Both colonels want to win as many fronts as possible. If one colonel knew how the other would arrange the troops, then this colonel would simply arrange to win all but one front by just one soldier. Thus, each colonel must follow a mixed strategy; that is, deciding at random how to arrange the troops. Of course, it matters just from which distribution is that randomly-determined pure strategy.

This paper studies Blotto on two battlefields where the outcome on a battlefield depends continuously on the advantage in allocations on that battlefield, a topic that has received scant attention to date. The motivation for this is that outnumbering one's opponent by one soldier does not guarantee victory. Having one more soldier is an advantage, but the battle might still go the other way. The bigger the advantage, the more likely the bigger army is to win. Thus, it is reasonable to have probability of winning be a continuous function of the two numbers of troops (instead of the function used in the classical Blotto game: 1 for the player with the bigger allocation, 0 for the other player). The goal is to maximize the expected value of the number of battlefields won. As expectation is a linear function, the probabilities of winning each battlefield can be re-interpreted as expected numbers of battlefields won. For example, a 57\% chance of winning would mean an expected value of 0.57 battlefields. This interpretation is all the more useful, because it can also accommodate varying degrees of victory. After all, not all battles are clear wins for one side or the other. It is also reasonable to subtract 0.5 from the function, so that a positive result on a battlefield represents an expected victory (the bigger the positive result, the bigger or more likely the victory), and a negative result, correspondingly, represents an expected defeat. Thus, an expected result of 0.57 battlefields becomes an expected result of 0.07 battlefields better than equality on that battlefield. This subtraction in effect turns a constant-sum game into a zero-sum game. This expected result can be written as a function of the difference of allocations on the battlefield, and this function can be called the outcome function.

If the two players had equal resources, then, at least if the outcome function were an odd function, then this game would have a simple solution: both players can play any strategy at all, and the result will be zero. Any gain on one battlefield would be compensated by a loss of equal size on the other battlefield. Therefore, it is necessary to introduce some unfairness into the game, either by allowing the outcome function not to be odd, or by giving one player an advantage in resources. The model of this paper accommodates both approaches. 

It is already known that a game like this has a Nash equilibrium; Glicksberg's theorem \cite{10.2307/2032478} guarantees it. This paper shows that if the outcome function is polynomial, then, not only is there a Nash equilibrium, but there is one where both players choose from only a small number of pure strategies. Actually finding an equilibrium exactly is a difficult matter.

This type of game can model two firms simultaneously competing on two markets, where one firm happens to be larger. It can also model two political parties competing in two districts (by spending campaign money) if (for instance) the number of seats for a party in a district is proportional to the vote percentage for that party in that district, something that classical win/loss Blotto cannot do. It can model winner-take-all rules also, because there is no guarantee of winning a district just by outspending the opponent. It is true that the exact outcome function is unlikely to be a polynomial, but it can be approximated by a polynomial fit through several points. The more points the polynomial goes through, the closer the result is to the Nash equilibrium, but the harder it is to evaluate the game.

In a zero-sum game, a Nash equilibrium represents perfect strategies for both players, in the sense that each player's strategy (if played as in the equilibrium) guarantees that player at least the equilibrium result. Those two combined guarantees rule out the possibility of either player seeking a better result - since the opponent's equilibrium strategy renders that impossible regardless of what the player may do. To the extent that a (zero-sum) Colonel Blotto game reflects an economic competition correctly, a Nash equilibrium to it gives perfect strategies for both players. Normally, with a continuous range of options, a Nash equilibrium looks like a continuous distribution, and all together, these constitute an infinite-dimensional space. This paper shows that a specific class of continuous Blotto games are exceptions to this, by having Nash equilibria that can be obtained by searching only a finite-dimensional space of strategies, each of which has a simple form.

The general approach that this paper uses is a series of reductions and transformations of the set of mixed strategies, which are repsesented as probability-density functions. First, any pure strategy of either player is represented by just a single number, $x$ for Player 1, and $y$ for Player 2. (Let Player 1 be the player with a resource advantage if there is one.) The payoff-matrix (call it $E$) based on this number has infinitely-many rows and columns (Player 1 choosing the row and Player 2 choosing the column).

However, of those infinitely-many columns (or rows), only $N+1$ (where $N$ is the degree of the outcome function) are linearly independent, essentially because the columns are polynomials of degree at most $N$ in $x$, Player 1's choice of pure strategy. This means that $E$ maps the space of all functions to a finite-dimensional space. So, there is only an $\left(N+1\right)$-dimensional space of functions $g$ (probability-density functions for Player 2's strategies) where all the $Eg$ are distinct. Thus, if $\tilde{g}$ is another function, then $E\tilde{g} = Eg$. But, if $f$ is the probability-density function of Player 1's mixed strategy (and $g$ is the same for Player 2), then the expected result is $f^T E g$ \footnote{$T$ stands for transpose, and this multiplication should be interpreted as if it were a usual matrix multiplication, with an integral replacing the summation.}. If $Eg = E\tilde{g}$, then $f^T E g = f^T E \tilde{g}$ regardless of what $f$ is, so every strategy for Player 2 is payoff equivalent to one of the $g$s in the $\left(N+1\right)$-dimensional space. This works just as well if the two players are reversed.

In other words, there is an $(N+1)$-dimensional space of equivalence classes of strategies (for either player), where any two strategies in the same equivalence class are payoff equivalent. Each of these equivalence classes can be described as an $(N+1)$-dimensional vector. That is true for both players.

What is more, the function that reduces the probability-density function of a strategy to the vector of the strategy's equivalence class can be made to be linear. Therefore, it preserves convex combinations\footnote{A convex combination is a weighted average with nonnegative weights; here, the weights are the probability-density function of the mixed strategy.}. Thus, the vectors of all the equivalence classes are convex combinations of the vectors of the equivalence classes of the pure strategies. The equivalence classes of the pure strategies form a curve in $(N+1)$-dimensional space. However, Carathéodory's Theorem says that each convex combination of points in $(N+1)$-dimensional space can be described as a convex combination of at most $N+2$ of the points. This means that every equivalence class that contains any strategies contains one that is a convex combination of only finitely many (at most $N+2$) pure strategies, or, equivalently, contains a discrete distribution with at most $N+2$ components\footnote{A component of a discrete-distribution mixed strategy is one of the pure strategies that are played with positive probability, according to the mixed strategy.}. Wherever any Nash-equilibrium strategy may be, its equivalence class contains a strategy with this type of discrete distribution, which is payoff equivalent to the Nash-equilibrium strategy (and thus is a Nash-equilibrium strategy too).

Finding Nash equilibria for Blotto games is the topic of many papers, each tasked with covering a specific variation of the conditions, the first one to propose the problem being Borel, in \cite{Borel1921} (translated into English as \cite{10.2307/1906946}). Things that were varied include the aim of the players (maximizing expected value or the probability of winning a majority), whether the allocations can be varied discretely or continuously, the number of battlefields fought on, the numbers of troops each colonel has access to (and whether one colonel has more troops than the other), the relative values of the fronts (where the aim is to maximize the total value won), whether the game is zero-sum or not, and even the nature of the resource constraints. These include Weinstein's \cite{RePEc:bpj:bejtec:v:12:y:2012:i:1:n:7}, Hart's \cite{Hart2008}, Gross and Wagner's \cite{2557533}, Roberson and Kvasov's \cite{Roberson2012}, Hortala-Vallve and Llorente-Saguer's \cite{Hortala-Vallve2012}, Schwartz, Loiseau, and Sastry's \cite{EURECOM+4479}, Thomas's \cite{RePEc:tex:wpaper:130116}, Kovenock and Roberson's \cite{RePEc:ces:ceswps:_5291}, and Macdonell and Mastronardi's \cite{Macdonell2015}. Macdonell and Mastronardi's paper has a similar result to this paper, in that it also has two battlefields with unequal resources, and a set of Nash equilibria that always includes a discrete distribution.

This paper explores yet another dimension of Blotto, where the result depends on the degree of victory. Even this form of Blotto was explored, although not as much as traditional win/loss Blotto. Unfortunately, these papers only search for pure strategies, and if there is none, the only conclusion that can be drawn from using the methods of those papers is exactly that: ''No pure-strategy Nash equilibrium exists''. One example is Blackett's \cite{NAV:NAV3800050203}, which only has a necessary condition for the existence of a pure-strategy Nash equilibrium. Blackett allows for any outcome function. Golman and Page's \cite{Golman2009} has a one-parameter family of outcome functions, and allows many battlefields and dependencies between battlefield outcomes, but Golman and Page find that in most of the cases they study, pure-strategy equilibria do not exist. Osório's \cite{RePEc:urv:wpaper:2072/211806} does succeed in finding a pure-strategy Nash equilibrium, but that paper restricts itself to a narrow class of outcome functions. This paper is a complement to the pure-strategy searches, in that it serves as a bound on how "impure" are the strategies that need to be considered. Mixed-strategy equilibria for this form of Blotto have been explored in \cite{Morozov2014}, but that paper limits its outcome function to a finite-parameter set.

The convex optimization method of Bellman's \cite{10.2307/2028140} is not applicable here, because of its assumption (in terms of the language of this paper) that, for all battlefields, the outcome function is concave down in Player 1's allocation and concave up in Player 2's allocation. When the outcome function is a function of the difference of allocations (as here), this assumption is almost always incorrect\footnote{The outcome function can be written as $P\left(\tilde{x}-\tilde{y}\right)$ (where $\tilde{x}$ and $\tilde{y}$ are the two players' allocations), and $\frac{\partial^2}{\partial \tilde{x}^2}\left[P\left(\tilde{x}-\tilde{y}\right)\right]$ and $\frac{\partial^2}{\partial \tilde{y}^2}\left[P\left(\tilde{x}-\tilde{y}\right)\right]$ both equal $P''\left(\tilde{x}-\tilde{y}\right)$, so the concavity of the outcome function is the same with respect to both players' allocations.}. 

However, a simpler version of the method of Beale and Heselden's \cite{NAV:NAV3800090202}, which is also a simpler version of the algorithm in Behnezhad, Dehghani, Derakhshan, HajiAghayi, and Seddighin's \cite{DBLP:journals/corr/BehnezhadDDHS16}, is applicable here. Player 1 takes an integer $L$ (the larger $L$, the better the approximation, but the longer it will take to calculate), and treats both players' possible allocations as if they were required to be integer multiple of $\frac{n+a}{L}$. Then, Player 1 considers the probabilities of playing each strategy as variables, adds an extra variable for the payoff, and sets up inequalities to reflect that the probabilities must be nonnegative and sum to 1, and that the expected payoff against any of Player 2's strategies be no lower than the payoff variable. Player 1 seeks to maximize the payoff variable subject to these inequalities, and uses a linear programming model to do this. This has the drawback of possibly producing increasingly complicated equilibrium strategies as $L$ goes up, which are proven here not to be necessary in the case of a polynomial outcome function. 

The knowledge that there exist Nash-equilibrium strategies that are discrete distributions, with an upper bound on the number of components, makes it easier to search for a Nash-equilibrium, analytically or numerically. One possible approach for this is to group each player's possible mixed strategies by the number of components, and by which ''edge'' strategies\footnote{An edge strategy is a pure strategy that is at a boundary of the pure-strategy space; here, this would mean placing all resources on one battlefield or the other.}, if any, were used as components. Then, in each group, the parameters that the player has control of are: which pure strategies are the components, and all but one of the probabilities of playing the components\footnote{The probability of playing the remaining component is determined by the fact that the probabilities have a sum of $1$, and so, it is not a parameter.}. For each pair of groups (one group for each player), the critical points of the payoff can be found by setting the partial derivatives of the payoff with respect to all parameters (for both players) to zero and solving the resulting system of polynomial equations.\footnote{If both players' groups have zero parameters (like the group of one-component strategies with an ''all-on-battlefield-2'' component), then every point is critical.} \footnote{There are no special cases for boundary extrema (from the point of view of one player), because these are already accounted for by being in a different group, either one with fewer components, or one using more "edge" strategies.} The Nash equilibrium is one of the critical points, and a point can be checked for being a Nash equilibrium by checking that neither player can improve the result by changing to a pure strategy. 

For efficiency, the critical points from the lowest-parameter groups should be checked first, before the systems for the higher-parameter groups get solved, as these are the groups that yield the simplest systems of equations. There are also ways to make this algorithm faster, by reducing the upper bound on the number of pure-strategy components that could be required. The extensions concern themselves primarily with this. One of the extensions takes advantage of the symmetry between the two battlefields, while another one uses the theorem proven in \cite{Part2} to take advantage of the continuity of the pure strategies. Together, these extensions make a brute-force approach (described in the extensions) feasible in practice for low degrees of the polynomial as the outcome function. 

\section{Model}
A two-field continuous Blotto game is defined by an ordered triple $\left(n, a, r\right)$, where $n$ is Player 2's resources, $a$ is Player 1's advantage in resources, and $r$ is the outcome function. That is, Player 1's resources are $n+a$, and Player 1 chooses a number $\tilde{x} \in \left[0, n+a\right]$. Player 2 chooses a number $\tilde{y} \in \left[0, n\right]$. These numbers are called allocations to battlefield 1. The allocation to battlefield 2 is $n+a - \tilde{x}$ for Player 1 and $n - \tilde{y}$ for Player 2. Thus, each player's allocations sum to that player's resources. 

On each battlefield, the outcome is $r\left(z\right)$, where $z$ is the difference in allocations between Player 1 and Player 2 on that battlefield. That is, $z = \tilde{x} - \tilde{y}$ on battlefield 1 and $z = a - \tilde{x} + \tilde{y}$ on battlefield 2. 

Player 1's payoff is $r\left(\tilde{x}-\tilde{y}\right) + r\left(a-\tilde{x}+\tilde{y}\right)$, and correspondingly, Player 2's payoff is the opposite, which $-r\left(\tilde{x}-\tilde{y}\right) - r\left(a-\tilde{x}+\tilde{y}\right)$, which makes continuous Blotto a zero-sum game. Player 1 seeks to maximize the expected value of Player 1's payoff, while Player 2 seeks to minimize the expected value of Player 1's payoff.

The outcome function represents the dependence of the degree of victory (or the probability of victory) on a battlefield on the advantage in resources on that battlefield, where a battlefield could (for instance) mean a district in an election, or one of two markets over which two firms simultaneously compete.

\section{Main result}

\textbf{Theorem.}
{\itshape In any continuous Blotto game where $r$ is a polynomial, there exists a (mixed-strategy) Nash-equilibrium, in which both players' strategies are distributions with support on at most $N+2$ points, where $N$ is the degree of $r$.}

It was already known that this game would have a Nash equilibrium; the new result is that there exists one with this form.

This theorem is proven by first establishing that, if $r$ is a polynomial, then the payoff matrix has rank not greater than $N+1$. That allows a change of coordinates that leaves both players with only an at most $\left(N+1\right)$-dimensional strategy space, each point in which corresponds to many mixed strategies that are exactly equivalent to each other. In the new coordinates, the mixed strategies still form the convex hull of the pure strategies, and that makes each mixed strategy have a representation as the convex combination of only at most $N+2$ pure strategies. This equality in the new coordinates corresponds to equivalence as mixed strategies. Wherever a Nash equilibrium may be, both players can choose an equivalent strategy with only finitely many components.

\section{Method}

The reduction of the strategy space to discrete distributions will be illustrated with a relatively simple example, $r\left(z\right)=-z^3$. It should be noted that when actually applying this result, these steps are not needed; they only exist to show that, in fact, there is a Nash equilibrium with a discrete distribution.

\subsection{Symmetrization of the Battlefields}
\label{shiftsec}
The players' possible strategies can equally well be written in terms of deviations from the even-split strategy. That is, instead of choosing $\tilde{x}$, Player 1 can be considered as choosing $x$, defined as $\tilde{x} - \frac{n+a}{2}$, and similarly, Player 2 can be considered as choosing $y$, defined as $\tilde{y} - \frac{n}{2}$. From the fact that $\tilde{x} \in \left[0, n+a\right]$, it follows that $x \in \left[-\frac{n+a}{2}, \frac{n+2}{2}\right]$, and from the fact that $\tilde{y} \in \left[0, n\right]$, it follows that $y \in \left[-\frac{n}{2}, \frac{n}{2}\right]$. 

Player 1's payoff, denoted by $R\left(x, y\right)$, is $r\left(x-y + \frac{a}{2}\right) + r\left(-x +y + \frac{a}{2}\right)$. The proof of this is in Appendix 1, as the proof of Lemma \ref{shiftlemma}. With $r\left(z\right)=-z^3$, $R\left(x, y\right) = -6\left(x-y\right)^2\left(\frac{a}{2}\right) - 2\left(\frac{a}{2}\right)^3$ (as demonstrated in Calculation \ref{calcR} in Appendix 2).

The goal of this change of variables is simply to treat the two battlefields symmetrically. Another useful property is that either player exchanging allocations to the two battlefields results in that player's variable ($x$ or $y$) simply changing sign when written in this form (because, since Player 1's allocations add to $n+a$, subtracting $\frac{n+a}{2}$ from each allocation leaves 0 for Player 1's new sum, and similarly, since Player 2's allocations add to $n$, subtracting $\frac{n}{2}$ from each allocation leaves 0 for Player 2's new sum).

Player 1's resource advantage gives Player 1 two benefits, namely the ability to have higher allocations and a wider range of strategies. This remains true after the change in coordinates, except that the advantages are decoupled from each other: the ability to have higher allocations is now in the form of the function $R$, while the wider range of strategies still manifests itself as a wider interval.

The addition of a constant to every payoff does not change which strategiy pairs are Nash equilibria, and neither does multiplying every payoff by a positive constant. As such, if the payoff were $-\left(x-y\right)^2$ instead of $-6\left(x-y\right)^2\left(\frac{a}{2}\right) - 2\left(\frac{a}{2}\right)^3$, that would not change the location of the Nash equilibria. Since the function $-\left(x-y\right)^2$ is easier to manipulate, that shall be used as the example for $R\left(x, y\right)$ in the rest of this paper.

\subsection{Expected Payoff as a Dot Product}
\label{payoffdotsec}

Let Player 1 and Player 2 both play mixed strategies. 

Player 1's mixed strategies are distributions of $x$, with support $\left[-\frac{n+a}{2}, \frac{n+a}{2}\right]$. Let $f\left(x\right)$ be the probability-density function of Player 1's mixed strategy. Similarly, Player 2's mixed strategies are distributions of $y$ with support $\left[-\frac{n}{2}, \frac{n}{2}\right]$. Let $g\left(y\right)$ be the probability-density function of Player 2's mixed strategy.

For functions $f_a$ and $f_b$ with domain $\left[-\frac{n+a}{2}, \frac{n+a}{2}\right]$, let $f_a \cdot f_b = \int \limits _{-\frac{n+a}{2}} ^{\frac{n+a}{2}} {f_a\left(x\right)f_b\left(x\right) dx}$. For functions $g_a$ and $g_b$ with domain $\left[-\frac{n}{2}, \frac{n}{2}\right]$, let $g_a \cdot g_b = \int \limits _{-\frac{n}{2}} ^{\frac{n}{2}} {g_a\left(y\right)g_b\left(y\right) dy}$. Let $E$ be the linear transformation, from the functions with domain $\left[-\frac{n}{2}, \frac{n}{2}\right]$ to the functions with domain $\left[-\frac{n+a}{2}, \frac{n+a}{2}\right]$, such that $Eg = \int \limits _{-\frac{n}{2}} ^{\frac{n}{2}} {R\left(x, y\right)g\left(y\right)dy}$.

Player 1's expected payoff in this game is the weighted average of all the possible results (with $f\left(x\right)$ and $g\left(y\right)$ being the weights), which is $\int \limits _{-\frac{n+a}{2}} ^{\frac{n+a}{2}} {\int \limits _{-\frac{n}{2}} ^{\frac{n}{2}} {R\left(x, y\right)f\left(x\right)g\left(y\right) dy} \,dx}$, can also be written as $f \cdot \left(Eg\right)$. (This is proven in Lemma \ref{expvallemma} in Appendix 1.)

If $R\left(x, y\right) = -\left(x - y\right)^2$, then:
\[
f \cdot \left(Eg\right) = \int \limits _{-\frac{n+a}{2}} ^{\frac{n+a}{2}} {\int \limits _{-\frac{n}{2}} ^{\frac{n}{2}} {\left(-\left(x-y\right)^2\right)f\left(x\right)g\left(y\right) dy} \,dx} 
\]

\subsection{Reduction of Dimensionality}
\label{dimreductionsec}

One way to analyze a linear transformation is to find its matrix, given a basis. This will be done with $R\left(x, y\right) = -\left(x-y\right)^2$, and the basis used will be the orthonormal basis of polynomials
\begin{align*}
& f_0\left(x\right) = \left(\sqrt{\frac{1}{2}}\right){\left(\frac{n+a}{2}\right)}^{-\frac{1}{2}} \\
& f_1\left(x\right) = \left(\sqrt{\frac{3}{2}}\right){\left(\frac{n+a}{2}\right)}^{-\frac{3}{2}} x \\
& f_2\left(x\right) = \left(\sqrt{\frac{45}{8}}\right){\left(\frac{n+a}{2}\right)}^{-\frac{5}{2}}x^2 - \left(\sqrt{\frac{5}{8}}\right){\left(\frac{n+a}{2}\right)}^{-\frac{1}{2}} \\
& f_3\left(x\right) = \left(\sqrt{\frac{175}{8}}\right){\left(\frac{n+a}{2}\right)}^{-\frac{7}{2}}x^3 - \left(\sqrt{\frac{63}{8}}\right){\left(\frac{n+a}{2}\right)}^{-\frac{3}{2}}x \\
& f_4\left(x\right) = \left(\sqrt{\frac{11025}{128}}\right){\left(\frac{n+a}{2}\right)}^{-\frac{9}{2}}x^4 - \left(\sqrt{\frac{2025}{32}}\right){\left(\frac{n+a}{2}\right)}^{-\frac{5}{2}} x^2 + \left(\sqrt{\frac{81}{128}}\right){\left(\frac{n+a}{2}\right)}^{-\frac{1}{2}} \\
& \cdots
\end{align*}
for Player 1, and the orthonormal basis of polynomials
\begin{align*}
& g_0\left(y\right) = \left(\sqrt{\frac{1}{2}}\right){\left(\frac{n}{2}\right)}^{-\frac{1}{2}} \\
& g_1\left(y\right) = \left(\sqrt{\frac{3}{2}}\right){\left(\frac{n}{2}\right)}^{-\frac{3}{2}} y \\
& g_2\left(y\right) = \left(\sqrt{\frac{45}{8}}\right){\left(\frac{n}{2}\right)}^{-\frac{5}{2}}y^2 - \left(\sqrt{\frac{5}{8}}\right){\left(\frac{n}{2}\right)}^{-\frac{1}{2}} \\
& g_3\left(y\right) = \left(\sqrt{\frac{175}{8}}\right){\left(\frac{n}{2}\right)}^{-\frac{7}{2}}y^3 - \left(\sqrt{\frac{63}{8}}\right){\left(\frac{n}{2}\right)}^{-\frac{3}{2}}y \\
& g_4\left(y\right) = \left(\sqrt{\frac{11025}{128}}\right){\left(\frac{n}{2}\right)}^{-\frac{9}{2}}y^4 - \left(\sqrt{\frac{2025}{32}}\right){\left(\frac{n}{2}\right)}^{-\frac{5}{2}} y^2 + \left(\sqrt{\frac{81}{128}}\right){\left(\frac{n}{2}\right)}^{-\frac{1}{2}} \\
& \cdots
\end{align*}
for Player 2. (See Calculation \ref{calcortho} in Appendix 2 for the process of obtaining it.)

Using this basis, $f$ can be written as
$
\left( 
\begin{matrix}
f \cdot f_{0} \\
f \cdot f_{1} \\
f \cdot f_{2} \\
... \\
\end{matrix}
\right)
$
and $g$ can be written as
$
\left( 
\begin{matrix}
g \cdot g_{0} \\
g \cdot g_{1} \\
g \cdot g_{2} \\
... \\
\end{matrix}
\right)
$. \footnote{This is true because $f \cdot f_s$ is also the $s$-th component of $f$, as, if $f = \sum \limits _{s = 0} ^{\infty} {c_s f_s}$ for some constants $c_s$, then $f \cdot f_s = f \cdot \sum \limits _{s = 0} ^{\infty} {c_s f_s} = \sum \limits _{s = 0} ^{\infty} {c_s \left(f \cdot f_s\right)}$, but as the $f_s$ are orthonormal, this is just $c_s$, which is the $s$-th component of $f$.} As for the matrix of $E$, its entry in the $i$-th row and $j$-th column should be the $i$-th component of $Eg_{j}$, but this is just $f_i \cdot \left(Eg_j\right)$. Thus, it follows that the matrix of $E$ is
\[
\left( 
\begin{matrix}
f_0 \cdot \left(Eg_0\right) & f_0 \cdot \left(Eg_1\right) & f_0 \cdot \left(Eg_2\right) & \cdots \\
f_1 \cdot \left(Eg_0\right) & f_1 \cdot \left(Eg_1\right) & f_1 \cdot \left(Eg_2\right) & \cdots \\
f_2 \cdot \left(Eg_0\right) & f_2 \cdot \left(Eg_1\right) & f_2 \cdot \left(Eg_2\right) & \cdots \\
\cdots & \cdots & \cdots & \cdots \\
\end{matrix}
\right)
\]
which means that $f \cdot \left(Eg\right)$ is
\[
\left( 
\begin{matrix}
f \cdot f_{0} \\
f \cdot f_{1} \\
f \cdot f_{2} \\
... \\
\end{matrix}
\right)
^T
\left( 
\begin{matrix}
f_0 \cdot \left(Eg_0\right) & f_0 \cdot \left(Eg_1\right) & f_0 \cdot \left(Eg_2\right) & \cdots \\
f_1 \cdot \left(Eg_0\right) & f_1 \cdot \left(Eg_1\right) & f_1 \cdot \left(Eg_2\right) & \cdots \\
f_2 \cdot \left(Eg_0\right) & f_2 \cdot \left(Eg_1\right) & f_2 \cdot \left(Eg_2\right) & \cdots \\
\cdots & \cdots & \cdots & \cdots \\
\end{matrix}
\right)
\left( 
\begin{matrix}
g \cdot g_{0} \\
g \cdot g_{1} \\
g \cdot g_{2} \\
... \\
\end{matrix}
\right)
\]
where $T$ indicates transpose. Using $f \cdot \left(Eg\right) = \int \limits _{-\frac{n+a}{2}} ^{\frac{n+a}{2}} {\int \limits _{-\frac{n}{2}} ^{\frac{n}{2}} {\left(-\left(x-y\right)^2\right)f\left(x\right)g\left(y\right) dy} \,dx}$, the matrix of $E$ is
\[
\hspace{-0.6 in}
\left(
\begin{matrix}
\left(-\frac{2}{3}\right)\left(\frac{n+a}{2}\right)^{\frac{5}{2}}\left(\frac{n}{2}\right)^{\frac{1}{2}} + \left(-\frac{2}{3}\right)\left(\frac{n+a}{2}\right)^{\frac{5}{2}}\left(\frac{n}{2}\right)^{\frac{1}{2}} & 0 & \left(-\frac{4\sqrt{5}}{15}\right)\left(\frac{n+a}{2}\right)^{\frac{1}{2}}\left(\frac{n}{2}\right)^{\frac{5}{2}} & 0 & 0 & \cdots \\
0 & \left(-\frac{4}{3}\right)\left(\frac{n+a}{2}\right)^{\frac{3}{2}}\left(\frac{n}{2}\right)^{\frac{3}{2}} & 0 & 0 & 0 & \cdots \\
\left(-\frac{4\sqrt{5}}{15}\right)\left(\frac{n+a}{2}\right)^{\frac{5}{2}}\left(\frac{n}{2}\right)^{\frac{1}{2}} & 0 & 0 & 0 & 0 & \cdots \\
0 & 0 & 0 & 0 & 0 & \cdots \\
0 & 0 & 0 & 0 & 0 & \cdots \\
\cdots & \cdots & \cdots & \cdots & \cdots & \cdots \\
\end{matrix}
\right)
\]
Note the all-zero columns and rows. Calculation \ref{calckernel} shows that those and all further columns and rows are, in fact, all zero.

The fact that only the top $3$-by-$3$ corner of the matrix has any non-zero entries means that $f \cdot \left(Eg\right)$ only depends on the first three rows and columns, and therefore, can be truncated to
\[
\hspace{-0.8 in}
\left( 
\begin{matrix}
f \cdot f_{0} \\
f \cdot f_{1} \\
f \cdot f_{2} \\
\end{matrix}
\right)
^T
\left(
\begin{matrix}
\left(-\frac{2}{3}\right)\left(\frac{n+a}{2}\right)^{\frac{5}{2}}\left(\frac{n}{2}\right)^{\frac{1}{2}} + \left(-\frac{2}{3}\right)\left(\frac{n+a}{2}\right)^{\frac{5}{2}}\left(\frac{n}{2}\right)^{\frac{1}{2}} & 0 & \left(-\frac{4\sqrt{5}}{15}\right)\left(\frac{n+a}{2}\right)^{\frac{1}{2}}\left(\frac{n}{2}\right)^{\frac{5}{2}} \\
0 & \left(-\frac{4}{3}\right)\left(\frac{n+a}{2}\right)^{\frac{3}{2}}\left(\frac{n}{2}\right)^{\frac{3}{2}} & 0 \\
\left(-\frac{4\sqrt{5}}{15}\right)\left(\frac{n+a}{2}\right)^{\frac{5}{2}}\left(\frac{n}{2}\right)^{\frac{1}{2}} & 0 & 0 \\
\end{matrix}
\right)
\left( 
\begin{matrix}
g \cdot g_0 \\
g \cdot g_1 \\
g \cdot g_2 \\
\end{matrix}
\right)
\]

The same is true in general. $f \cdot \left(Eg\right)$ still equals 
\[
\left( 
\begin{matrix}
f \cdot f_{0} \\
f \cdot f_{1} \\
f \cdot f_{2} \\
... \\
\end{matrix}
\right)
^T
\left( 
\begin{matrix}
f_0 \cdot \left(Eg_0\right) & f_0 \cdot \left(Eg_1\right) & f_0 \cdot \left(Eg_2\right) & \cdots \\
f_1 \cdot \left(Eg_0\right) & f_1 \cdot \left(Eg_1\right) & f_1 \cdot \left(Eg_2\right) & \cdots \\
f_2 \cdot \left(Eg_0\right) & f_2 \cdot \left(Eg_1\right) & f_2 \cdot \left(Eg_2\right) & \cdots \\
\cdots & \cdots & \cdots & \cdots \\
\end{matrix}
\right)
\left( 
\begin{matrix}
g \cdot g_{0} \\
g \cdot g_{1} \\
g \cdot g_{2} \\
... \\
\end{matrix}
\right)
\]
and this can be truncated to at most $N+1$ dimensions (rows and columns), where $N$ is the degree of $r$ as a polynomial. ($f_0, f_1, f_2, \cdots$ and $g_0, g_1, g_2, \cdots$ are exactly the same orthonormal polynomials as they were in the $R\left(x, y\right) = -\left(x-y\right)^2$ case.) This statement is proven in Appendix 1 as Lemma \ref{finitelemma}.

Let the old coordinates refer to the coordinates where  Player 1's strategy is represented by its probability-density function $f$, and Player 2's strategy is represented by its probability-density function $g$, and let the new coordinates be 
$
\left( 
\begin{matrix}
f \cdot f_{0} \\
f \cdot f_{1} \\
f \cdot f_{2} \\
\cdots \\
f \cdot f_{N}
\end{matrix}
\right)
$
 for Player 1 and  
$
\left( 
\begin{matrix}
g \cdot g_{0} \\
g \cdot g_{1} \\
g \cdot g_{2} \\
\cdots \\
g \cdot g_{N}
\end{matrix}
\right)
$
 for Player 2, where $N$ is the degree of $r$ as a polynomial.

\subsection{Convex Hull}
\label{convexsec}

While the strategy space in the new coordinates is finite-dimensional if $r$ is a polynomial, its boundaries are rather difficult to find directly. However, tools from convexity theory allow an easier description of the strategy space.

Every mixed strategy is a convex combination of pure strategies, in that if $g$ is a mixed strategy, then $g\left(y\right) = \int \limits _{-\frac{n}{2}} ^{\frac{n}{2}} {\delta\left(t_y - y\right)g\left(t_y\right)dt_y}$, where $\delta\left(t_y - y \right)$ is the Dirac delta function applied to $t_y - y$, which is the probability density function of a random variable that always takes on the value $t_y$. Since the transformation to the new coordinates is linear, it follows that every mixed strategy in the new coordinates is likewise a convex combination of the pure strategies. In fact, the coefficients that show the mixed strategy is a convex combination of the pure strategies are exactly the same as the coefficients that do this in the old coordinates. 

Thus, the convex hull of the pure streategies in the new coordinates contains all the mixed strategies. If $r$ is a polynomial of degree $N$, the pure strategies are part of an (at most) $N+1$-dimensional space. Lemma \ref{convexlemma} in Appendix 1 shows that every point in the convex hull of an $N+1$-dimensional set is a convex combination of at most $N+2$ points, so every mixed strategy is a convex combination of $N+2$ or fewer strategies. This includes whatever the Nash equilibrium is.

In the example with $R\left(x, y\right) = \left(x - y\right)^2$, the strategies, in the new coordinates, are 
$
\left(
\begin{matrix}
f \cdot f_0 \\
f \cdot f_1 \\
f \cdot f_2
\end{matrix}
\right)
$ for Player 1 and 
$
\left(
\begin{matrix}
g \cdot g_0 \\
g \cdot g_1 \\
g \cdot g_2
\end{matrix}
\right)
$
 for Player 2.

The pure strategies are, in the new coordinates (as shown in Calculation \ref{calcpurestrategies} in Appendix 2),
\[
\left(
\begin{matrix}
\left(\sqrt{\frac{1}{2}}\right){\left(\frac{n+a}{2}\right)}^{-\frac{1}{2}} \\
\left(\sqrt{\frac{3}{2}}\right){\left(\frac{n+a}{2}\right)}^{-\frac{3}{2}} x \\
\left(\sqrt{\frac{45}{8}}\right){\left(\frac{n+a}{2}\right)}^{-\frac{5}{2}}x^2 - \left(\sqrt{\frac{5}{8}}\right){\left(\frac{n+a}{2}\right)}^{-\frac{1}{2}} \\
\end{matrix}
\right)
\]
for Player 1, and
\[
\left(
\begin{matrix}
\left(\sqrt{\frac{1}{2}}\right){\left(\frac{n}{2}\right)}^{-\frac{1}{2}} \\
\left(\sqrt{\frac{3}{2}}\right){\left(\frac{n}{2}\right)}^{-\frac{3}{2}} y \\
\left(\sqrt{\frac{45}{8}}\right){\left(\frac{n}{2}\right)}^{-\frac{5}{2}}y^2 - \left(\sqrt{\frac{5}{8}}\right){\left(\frac{n}{2}\right)}^{-\frac{1}{2}} \\
\end{matrix}
\right)
\]
for Player 2.
This can be viewed as a parametric representation of the curve containing all pure strategies (in the new coordinates). The Nash equilibrium strategies (for both players) must lie within the convex hull of this curve, and Calculation \ref{symmetrycalc} in Appendix 2 reduces the strategy space further, to a one-dimensional one (and actually finds the Nash equilibrium from there, although it may be easier to determine the strategy space in terms of the old coordinates, and then find the Nash equilibrium; after all, in the old coordinates, these strategies have simple descriptions).

This is fewer than the 4 strategies (which imply a 7-parameter family of strategies: 4 strategies, 4 coefficients, minus 1 for the coefficients always summing to 1) allowed by Lemma \ref{convexlemma} in Appendix 1 (and the overall theorem allows 5 strategies, which give 9 parmaters). Moreover, in each individual mixed strategy, the two strategies have fixed weights ($\frac{1}{2}$ each), so in fact, one parameter suffices for naming each mixed strategy that could be a Nash equilibrium. Are there always going to be fewer parameters describing the space in which a Nash equilibrium must lie? The answer is yes, and some of the extensions concern themselves with this.

\section{Extensions}

\subsection{Reducing the Number of Components}

\subsubsection{Degree of $R$}

$R\left(x, y\right) = r\left(x - y +\frac{a}{2}\right) + r\left(-x + y + \frac{a}{2}\right)$. Therefore, $R$ is an even function of $x - y$, which means that its degree in $x - y$ is even. Thus, if $r\left(z\right)$ has degree $N$ where $N$ is odd, then $R\left(x, y\right)$ cannot have degree $N$. At most, $R\left(x, y\right)$ has degree $N - 1$. That reduces the maximum number of components that might be required to describe a Nash-equilibrium strategy by one if $N$ is odd, so if $N$ is odd, there is a Nash-equilibrium strategy with only $N+1$ components (or fewer), rather than $N+2$.

\subsubsection{Removing the Odd Coordinates}
\label{componentsec}

The equivalence of the two battlefields means that only symmetrical strategies need to be considered, where a symmetrical strategy is one that, for whatever probability (or probability density) it assigns to a particular pure strategy, it assigns the same probability (or probability density) to the same strategy with the battlefields reversed. As exchanging allocations turns a pure strategy $x$ into pure strategy $-x$ (or $y$ into $-y$), if Player 1 exchanges allocations in every pure strategy of which Player 1's mixed strategy $f\left(x\right)$ is composed, that turns Player 1's strategy into $f\left(-x\right)$. A similar allocation exchange for Player 2 turns $g\left(y\right)$ into $g\left(-y\right)$. That makes any symmetrical strategy be represented by an even function.

A strategy can be symmetrized by taking its even part, $f_e\left(x\right) = \frac{f\left(x\right)+f\left(-x\right)}{2}$ (or $g_e\left(y\right) = \frac{g\left(y\right)+g\left(-y\right)}{2}$. Since this is a mixture of two mixed strategies with coefficients adding to 1, it is likewise a mixed strategy. If one player plays a symmetrical strategy, then the other player can symmetrize, without affecting the result, because a symmetrical strategy is invariant to the opponent exchanging allocations (and symmetrizing is mixing the original strategy with the same strategy with exchanged allocations). If Player 1 playing strategy $f$ and Player 2 playing strategy $g$ is a Nash equilibrium, then none of the steps of the following cycle: Player 1 symmetrizing, then Player 2 symmetrizing, then Player 1 removing the symmetrization, then Player 2 removing the symmetrization, improves Player 1's payoff. That means that each step keeps the payoff the same, so there is no effect on the payoff from both players symmetrizing from a Nash equilibrium. Furthermore, if a deviation from $f_e$ succeeded (in improving the result against the original, with Player 1 playing strategy $f_e$ and Player 2 playing strategy $g_e$) against $g_e$, the symmetrization of the deviation would succeed equally well, and the symmetrization would also succeed against $g$, but that is impossible because Player 1 playing strategy $f$ and Player 2 playing strategy $g$ is a Nash equilibrium. The deviation does not succeed for Player 1, and a deviation by Player 2 would similarly not succeed. Thus, if Player 1 playing strategy $f$ and Player 2 playing strategy $g$ is a Nash equilibrium, so is Player 1 playing strategy $f_e$ and Player 2 playing strategy $g_e$. More details can be found in Lemma \ref{evenequillemma} in Appendix 1 and its proof. 

Lemma \ref{orthobasislemma} in Appendix 1 confirms that $f_0, f_1, f_2, \cdots$ and $g_0, g_1, g_2, \cdots$ are all even or odd functions, and therefore, the new coordinates need only include the components corresponding to even-function polynomials, which  are polynomials with only even powers, the components of an even function corresponding to odd-function polynomials being zero. That reduces the dimension of the mixed strategy space to $\frac{N+2}{2}$ (or to $\frac{N+1}{2}$ if $N$ is odd), this being the number of nonnegative even integers less than or equal to the degree of $R$. As symmetrizing is a linear operator, these mixed strategies are all and only the convex combinations of the symmetrized pure strategies. (A symmetrized pure strategy at $x$, or $y$, is a mixed strategy of $x$ and $-x$, or correspondingly, $y$ and $-y$, each with probability $\frac{1}{2}$.) That means that there is a Nash-equilibrium strategy that is a convex combination of only $\frac{N+4}{2}$ (or $\frac{N+3}{2}$ if $N$ is odd) symmetrized pure strategies.

In the new coordinates, the strategies were $
\left(
\begin{matrix}
f \cdot f_0 \\
f \cdot f_1 \\
f \cdot f_2 \\
\cdots \\
f \cdot f_N
\end{matrix}
\right)
$ for Player 1 and 
$
\left(
\begin{matrix}
g \cdot g_0 \\
g \cdot g_1 \\
g \cdot g_2 \\
\cdots \\
g \cdot g_N
\end{matrix}
\right)
$
 for Player 2. This consideration reduces the strategy space to $
\left(
\begin{matrix}
f \cdot f_0 \\
f \cdot f_2 \\
f \cdot f_4 \\
\cdots \\
f \cdot f_{2\lfloor \frac{N}{2}\rfloor }
\end{matrix}
\right)
$ for Player 1 and 
$
\left(
\begin{matrix}
g \cdot g_0 \\
g \cdot g_2 \\
g \cdot g_4 \\
\cdots \\
g \cdot g_{2\lfloor \frac{N}{2}\rfloor }
\end{matrix}
\right)
$ for Player 2.\footnote{The subscript of the last entry is $N$ if $N$ is even, or $N-1$ if $N$ is odd.} \footnote{The functions with the even subscripts are even, and the functions with the odd subscripts are odd, because one term of $f_s$ as a polynomial in $x$ is a nonzero multiple of $x^s$, and therefore, as $f_s$ was already shown to be even or odd, it must be even if $x^s$ is even, or odd if $x^s$ is odd. However, $x^s$ is even exactly when $s$ is even, and $x^s$ is odd exactly when $s$ is odd.}

\subsubsection{Removing the $0$-th Coordinate}

When Carathéodory's Theorem was used, the dimension of the curve of pure strategies equaled the number of parameters. However, this curve actually has a dimension one lower than the number of parameters, because $f \cdot f_0$ and $g \cdot g_0$ are constants, and that is because 
\begin{align*}
f \cdot f_0 &= \int \limits _{-\frac{n+a}{2}} ^{\frac{n+a}{2}} {f\left(x\right) f_0\left(x\right) dx} \\
&= \int \limits _{-\frac{n+a}{2}} ^{\frac{n+a}{2}} {f\left(x\right) \left(\left(\sqrt{\frac{1}{2}}\right)\left(\frac{n+a}{2}\right)^{-\frac{1}{2}}\right) dx} \\
&= \left(\left(\sqrt{\frac{1}{2}}\right)\left(\frac{n+a}{2}\right)^{-\frac{1}{2}}\right)\int \limits _{-\frac{n+a}{2}} ^{\frac{n+a}{2}} {f\left(x\right) dx} \\
&= \left(\sqrt{\frac{1}{2}}\right)\left(\frac{n+a}{2}\right)^{-\frac{1}{2}}
\end{align*}
(because $\int \limits _{-\frac{n+a}{2}} ^{\frac{n+a}{2}} {f\left(x\right) dx} = 1$), and similarly for $g \cdot g_0$. This reduces the true dimension of either player's strategy space by one.

\subsubsection{Extending Carathéodory's Theorem}

Another extension is to make use of the fact that the set of all pure strategies is a polynomial curve. Carathéodory's Theorem assumes the worst-case scenario, that is, that the points of which the convex hull is taken are discrete. This is not the case here; the locations of the pure strategies can be varied continuously, so, if the curve is $M$-dimensional, then, instead of the usual $M+1$ points on the  curve being required to name any point in the convex hull, only $\frac{M+1}{2}$ points are needed, as each component contributes $2$ to the dimension of the space of the strategies: $1$ for the probability of playing this pure strategy, plus $1$ for the location of the pure strategy.\footnote{A number of points that is a half-integer is to be interpreted as the next higher integer number of points, with the first point fixed at its lowest possible value. This is ``half a point'' in the sense that it only contributes $1$ to the dimension of the space of the strategies, instead of $2$.} That this is indeed possible is shown as the corollary of the theorem in \cite{Part2}.

\subsubsection{The Total Effect}

These dimension-reducing extensions are not mutually-exclusive; indeed, all of them can be applied, in the order in which they were presented. After the odd coordinates were removed, the strategy space became
$
\left(
\begin{matrix}
f \cdot f_0 \\
f \cdot f_2 \\
f \cdot f_4 \\
\cdots \\
f \cdot f_{2\lfloor \frac{N}{2}\rfloor }
\end{matrix}
\right)
$ for Player 1 and 
$
\left(
\begin{matrix}
g \cdot g_0 \\
g \cdot g_2 \\
g \cdot g_4 \\
\cdots \\
g \cdot g_{2\lfloor \frac{N}{2}\rfloor }
\end{matrix}
\right)
$ for Player 2, and after the $0$-th coordinate was removed, the strategy space became only $\lfloor \frac{N}{2} \rfloor$-dimensional, and the allowance of $2$ parameters per component means that only $\frac{\lfloor \frac{N}{2} \rfloor+1}{2}$ symmetrized pure strategies are required. Thus, only $\frac{N+2}{4}$ if $N$ is even, or $\frac{N+1}{4}$ if $N$ is odd, symmetrized pure strategies are required to find the Nash equilibrium.

This low number allows the use of an otherwise-infeasible brute-force algorithm to find an approximate Nash equilibrium. First, find the payoff when both players play symmetrized pure strategies, in terms of which symmetrized pure strategies the players play. Then, restrict Player 1's symmetrized pure strategies to $\left \lbrace 0, \frac{\left(\frac{n+a}{2}\right)}{L}, 2\frac{\left(\frac{n+a}{2}\right)}{L}, \cdots \frac{n+a}{2} \right \rbrace$, for some integer $L$. Do the same for Player 2, which gives the set $\left \lbrace 0, \frac{\left(\frac{n}{2}\right)}{L}, 2\frac{\left(\frac{n}{2}\right)}{L}, \cdots \frac{n}{2} \right \rbrace$. For (symmetrized) mixed strategies, express them as convex combinations of symmetrized pure strategies, and restrict the coefficients to the set of integer multiples of $\frac{1}{L}$ (with these coefficients still summing to $1$). Then, for each mixed strategy\footnote{The only mixed strategies that are to be used here are those that are convex combinations of up to $\frac{N+2}{4}$, if $N$ is even, (or of up to $\frac{N+1}{4}$, if $N$ is odd) symmetrized pure strategies.} of Player 1, and for each symmetrized pure strategy of Player 2 (for both players, using the only pure strategies possible to be ones in the ``restricted'' set), compute the expected payoff. For each mixed strategy of Player 1, choose the lowest result obtained by Player 2's pure strategies. Over all the mixed strategies, the highest such minimum is the approxmate Nash-equilibrium strategy for Player 1. Reverse the players (so now Player 2 selects a mixed strategy, while Player 1 selects a symmetrized pure strategy) to get an approximate Nash-equilibrium strategy for Player 2. A check on the accuracy of this can be obtained by finding the best result either player can get by playing any (not just one in the ``restricted'' set) pure strategy against the other player's claimed approximate Nash-equilibrium strategy.

\subsection{Other Possible Extensions}

The main result need not be limited to the class of games examined here; for instance, it would also apply if the choices of allocations for both players were discrete, or if the two battlefields did not have the same outcome function, or if the outcome function (on either battlefield) were any polynomial of degree $N$ or less in both $x$ and $y$ (not just a polynomial of degree $N$ or less in $x-y$), or any combination of these. However, not all of these extensions apply in all of these cases.

Polynomials are not the only kind of function that give a payoff matrix of finite rank. This theorem might be modified to be applied to other functions producing payoff matrices of finite rank, perhaps functions like $P\left(z\right)\sin\left(z\right)$ or $P\left(z\right)\sinh\left(z\right)$, where $P\left(z\right)$ is a polynomial.

With more than two batlefields, the strategy spaces still have an orthonormal basis, this time of polynomials in several variables. Might a Blotto game with the same polynomial outcome function on each battlefield also have a Nash equilibrium with a discrete distribution? If so, how many components does this discrete distribution have? Perhaps, the methods of this paper can be used to answer these questions.

\section{Conclusion}

Modeling a competition using Blotto with a polynomial outcome function has two advantages over modeling the same competition with Blotto where the outcome on every battlefield is a win or a loss. One advantage is that a polynomial is a better fit to the actual function than a staircase function, especially if the staircase function is required to have its jump at zero (which it should if we want to make the battlefield fair). This is even true if the battlefields actually do have a winner-take-all rule, because a resource advantage does not guarantee victory. 

The other advantage is the existence of a Nash-equilibrium strategy that is a discrete distribution with relatively few components. This leaves only a finite-dimensional space of possible Nash equilibria, which can be searched for saddle points. This is also a good reason to try to minimize the dimension of the space in which Nash equilibria exist, so that the objects on which there is no small perturbation of strategy that improves the result are as low-dimensional as possible; ideally, these should be points; hence the extensions. Once such an equilibrium strategy is found, it is easy to implement with a random-number generator. The same cannot be said about continuous distributions in general, because simulating such a distribution requires knowing how to invert the cumulative distribution function, and this inverse might not have a convenient form.


\section{Appendix 1: Lemmas}

\begin{lem}[Effect of the shift of coordinates on the payoff]
\label{shiftlemma}
Player 1's payoff is $r\left(x-y + \frac{a}{2}\right) + r\left(-x +y + \frac{a}{2}\right)$.
\end{lem}

\begin{proof}
Player 1's payoff is:
\begin{align*}
& r\left(\tilde{x}-\tilde{y}\right) + r\left(a-\tilde{x}+\tilde{y}\right) \\
& = r\left(\left(x + \frac{n+a}{2}\right)-\left(y + \frac{n}{2}\right)\right) + r\left(a-\left(x + \frac{n+a}{2}\right)+\left(y + \frac{n}{2}\right)\right) \\
& = r\left(x + \frac{n+a}{2}-y - \frac{n}{2}\right) + r\left(a-x - \frac{n+a}{2}+y + \frac{n}{2}\right) \\
& = r\left(x-y + \frac{a}{2}\right) + r\left(-x +y + \frac{a}{2}\right)
\end{align*}
\end{proof}

\begin{lem}
\label{expvallemma}
Player 1's expected payoff, which is $\int \limits _{-\frac{n}{2}} ^{\frac{n}{2}} {\int \limits _{-\frac{n}{2}} ^{\frac{n}{2}} {R\left(x, y\right)f\left(x\right)g\left(y\right) dy} \,dx}$, equals $f \cdot \left(Eg\right)$.
\end{lem}

\begin{proof}
Player 1's expected payoff can be simplified as follows:
\begin{align*}
\int \limits _{-\frac{n}{2}} ^{\frac{n}{2}} {\int \limits _{-\frac{n}{2}} ^{\frac{n}{2}} {R\left(x, y\right)f\left(x\right)g\left(y\right) dy} \,dx} & = \int \limits _{-\frac{n}{2}} ^{\frac{n}{2}} {f\left(x\right)\int \limits _{-\frac{n}{2}} ^{\frac{n}{2}} {R\left(x, y\right)g\left(y\right) dy} \,dx} \\
& = \int \limits _{-\frac{n}{2}} ^{\frac{n}{2}} {\left(f\left(x\right)\right)\left( \left(Eg\right)\left(x\right)\right) dx} \\
& = f \cdot \left(Eg\right)
\end{align*}
and thus equals $f \cdot \left(Eg\right)$.
\end{proof}

\begin{lem}
\label{decomp}
If $g$ is a function, there is exactly one decomposition of $g$ into $g_{||} + g_{\perp}$, such that $g_{||}$ is a polynomial with degree at most $M$, and $g_{\perp}$ is orthogonal to all polynomials with degree at most $M$. 
\end{lem}
\emph{Proof sketch}. Use Gram-Schmidt orthogonalization to get an orthonormal basis of the polynomials with degree at most $M$. Then, project $g$ to the space of polynomials with degree at most $M$, for instance by adding the projections of $g$ to the spaces of each individual vector of the basis. This is $g_{||}$. What is left is orthogonal to all polynomials with degree at most $M$.

This is the only decomposition; if there were two distinct ones, subtract one from the other, which gives an alternative way to write $0$: $0 = 0_{||} + 0_{\perp}$. $0_{||} = -0_{\perp}$, so $0_{||}$ is both a polynomial of degree at most $M$ and orthogonal to all polynomials of degree at most $M$, including itself. That makes $0_{||} = 0$ and $0_{\perp} = 0$, so there is no alternative decomposition.
\begin{proof}
As the polynomials with degree at most $M$ form an $\left(M+1\right)$-dimensional vector space, this vector space must have an orthonormal basis (which can be constructed by the Gram-Schmidt process, starting from the non-orthonormal basis $\left \lbrace 1, y, y^2, ..., y^M \right \rbrace$). Let $\left \lbrace b_0, b_1, b_2, ..., b_M \right \rbrace$ be this orthonormal basis. Then, $g_{||} = \sum \limits _{t = 0} ^{M} {\left(\left(g \cdot b_t\right)b_t\right)}$, and $g_{\perp} = g - g_{||}$. Indeed, for any $b_u$ with $u \in \left \lbrace 0, 1, 2, ..., M \right \rbrace$:
\begin{align*}
b_u \cdot g_{\perp} &= b_u \cdot \left(g - g_{||} \right) \\
&= b_u \cdot \left(g - \left(\sum \limits _{t = 0} ^{M} {\left(\left(g \cdot b_t\right)b_t\right)}\right) \right) \\
&= b_u \cdot g - \sum \limits _{t = 0} ^{M} {\left(\left(g \cdot b_t\right)\left(b_u \cdot b_t\right)\right)} \\
&= b_u \cdot g - g \cdot b_u \\
& = 0 \\
\end{align*}

(As $\left \lbrace b_0, b_1, b_2, ..., b_M \right \rbrace$ is an orthonormal basis, all terms of $\sum \limits _{t = 0} ^{M} {\left(\left(g \cdot b_t\right)\left(b_u \cdot b_t\right)\right)}$ except for the term with $t = u$, and if $t = u$, then $b_t \cdot b_u = 1$.) 

 The decomposition is unique, because if there were two decompositions $g =  g_{|| 1} + g_{\perp 1} = g_{|| 2} + g_{\perp 2}$, then:
\begin{align*}
\left(g_{|| 1} + g_{\perp 1}\right) - \left( g_{|| 2} + g_{\perp 2}\right) &= 0 \\
\left(g_{|| 1} - g_{|| 2}\right) + \left(g_{\perp 1} - g_{\perp 2}\right) &= 0 \\
\left(g_{|| 1} - g_{|| 2}\right) &= -\left(g_{\perp 1} - g_{\perp 2}\right) \\
\end{align*}
However, this means that $g_{|| 1} - g_{|| 2}$ is both a polynomial with degree at most $M$, and is orthogonal to all polynomials with degree at most $M$. In particular, this means that $g_{|| 1} - g_{|| 2}$ is orthogonal to itself, and the only polynomial that is orthogonal to itself is the zero polynomial. The same is true of $g_{\perp 1} - g_{\perp 2}$. Thus, $g_{|| 1} = g_{|| 2}$ and $g_{\perp 1} = g_{\perp 2}$, so the two decompositions were the same.
\end{proof}

\begin{lem}[Finite Rank of $E$]
\label{finitelemma}

If $r\left(z\right)$ is a polynomial function of $z$ with degree $N$, if $f_0, f_1, \cdots, f_M$ are orthonormal and are all polynomials in $x$ with degree not greater than $M$, and if $g_0, g_1, \cdots, g_M$ are orthonormal and are all polynomials in $y$ with degree not greater than $M$, then, for some $M \leq N$, $f \cdot \left(Eg\right)$ can be expressed as 
\[
\left(
\begin{matrix}
f \cdot f_0 & f \cdot f_1 & \cdots & f \cdot f_M \\
\end{matrix}
\right)
\left(
\begin{matrix}
f_0 \cdot \left(Eg_0\right) & f_0 \cdot \left(Eg_1\right) & \cdots & f_0 \cdot \left(Eg_M\right) \\
f_1 \cdot \left(Eg_0\right) & f_1 \cdot \left(Eg_1\right) & \cdots & f_1 \cdot \left(Eg_M\right) \\
\cdots & \cdots & \cdots & \cdots \\
f_M \cdot \left(Eg_0\right) & f_M \cdot \left(Eg_1\right) & \cdots & f_M \cdot \left(Eg_M\right) \\
\end{matrix}
\right)
\left(
\begin{matrix}
g \cdot g_0 \\
g \cdot g_1 \\
\cdots \\
g \cdot g_M \\
\end{matrix}
\right)
\]

\end{lem}
\emph{Proof sketch}. Let $M$ be the degree of $R\left(x, y\right)$ as a polynomial in $x - y$. Decompose $f \cdot \left(Eg\right)$ into $f_{||} \cdot \left(Eg_{||}\right) + f_{\perp} \cdot \left(Eg_{||}\right) + f_{||} \cdot \left(Eg_{\perp}\right) + f_{\perp} \cdot \left(Eg_{\perp}\right)$. The last two terms are zero, because $R\left(x, y\right)$ is a polynomial of degree $M$ in both $x$ and $y$, and at each $x$, applying $E$ to $g_{\perp}$ is equivalent to taking the dot product of $R$ and $g_{\perp}$, which is zero. $Eg$ is a polynomial in $x$ with degree at most $M$ (as $R\left(x, y\right)$ is such, and as $Eg = \int \limits _{-\frac{n}{2}} ^{\frac{n}{2}} {R\left(x, y\right)g\left(y\right)dy}$ is an integral with respect to $y$), and $f_{\perp}$ is orthogonal to that, so the second term is zero also. Thus, $f \cdot \left(Eg\right) = f_{||} \cdot \left(Eg_{||}\right)$.

Decompose $f$ according to an orthonormal basis $\left \lbrace f_0, f_1, \cdots f_M \right \rbrace$, and decompose $g$ according to an orthonormal basis $\left \lbrace g_0, g_1, \cdots g_M \right \rbrace$. Then, if $f = f_{s_x}$ for some $s_x \in \left \lbrace 0, 1, \cdots, M \right \rbrace$ and $g = g_{s_y}$ for some $s_y \in \left \lbrace 0, 1, \cdots, M \right \rbrace$, then the equation ($f \cdot \left(Eg\right)$ equaling the expression) holds. Thus, using linear combinations, the equation holds for $f = f_{||}$ and $g = g_{||}$, because both sides are linear in $f$ and in $g$. However, for all $s \in \left \lbrace 0, 1, \cdots, M \right \rbrace$, $f_{\perp} \cdot f_s = 0$ and $g_{\perp} \cdot g_s = g_{||} \cdot g_s$, so, for $f = f_{\perp}$ or $g = g_{\perp}$, both sides of the equation are zero, so it holds again. Thus, again using linear combinations, the equation holds for $f$ in general.
\begin{proof}
\begin{align*}
Eg & = \int \limits _{-\frac{n}{2}} ^{\frac{n}{2}} {R\left(x, y\right)g\left(y\right)dy} \\
& = \int \limits _{-\frac{n}{2}} ^{\frac{n}{2}} {\left(r\left(x - y + \frac{a}{2}\right)+r\left(-x + y + \frac{a}{2}\right) \right)g\left(y\right)dy} \\
& = \int \limits _{-\frac{n}{2}} ^{\frac{n}{2}} {\left(r\left(x - y + \frac{a}{2}\right)+r\left(-\left(x - y\right) + \frac{a}{2}\right) \right)g\left(y\right)dy}
\end{align*}
$\frac{a}{2}$ being a constant, it follows that $r\left(x - y + \frac{a}{2}\right)+r\left(-\left(x - y\right) + \frac{a}{2}\right)$ is a polynomial in $x - y$ with degree not more than $N$. Let $P\left(x - y\right) = r\left(x - y + \frac{a}{2}\right)+r\left(-\left(x - y\right) + \frac{a}{2}\right)$, and let $M$ be the degree of $P$. This means that $M$ is a nonnegative integer not greater than $N$. Thus, $R\left(x, y\right) = P\left(x - y\right)$. 

Decompose $g$ into $g_{||} + g_{\perp}$, such that $g_{||}$ is a polynomial with degree at most $M$, and $g_{\perp}$ is orthogonal to all polynomials with degree at most $M$ (which can be done in exactly one way, by Lemma \ref{decomp}). Also decompose $f$ into $f_{||} + f_{\perp}$. Then:
\begin{align*}
f \cdot \left(Eg\right) &= f \cdot \left(\int \limits _{-\frac{n}{2}} ^{\frac{n}{2}} {P\left(x - y\right)g\left(y\right)dy}\right) \\
&= f \cdot \left(\int \limits _{-\frac{n}{2}} ^{\frac{n}{2}} {P\left(x - y\right)\left(g_{||}\left(y\right) + g_{\perp}\left(y\right)\right)dy }\right) \\
&= f \cdot \left(\int \limits _{-\frac{n}{2}} ^{\frac{n}{2}} {P\left(x - y\right)g_{||}\left(y\right)dy }+\int \limits _{-\frac{n}{2}} ^{\frac{n}{2}} {P\left(x - y\right)g_{\perp}\left(y\right)dy }\right) \\
\end{align*}
$P$ is a polynomial in $x-y$ with degree at most $M$, which means that, regardless of the value of $x$, $g_{\perp}$ is orthogonal to $P$. That means that $\int \limits _{-\frac{n}{2}} ^{\frac{n}{2}} {P\left(x - y\right)g_{\perp}\left(y\right) dy} = 0$ for all values of $x$ (in $\left[-\frac{n}{2},\frac{n}{2}\right]$). Thus:
\begin{align*}
f \cdot \left(Eg\right) &=  f \cdot \left(\int \limits _{-\frac{n}{2}} ^{\frac{n}{2}} {P\left(x - y\right)g_{||}\left(y\right)dy }+\int \limits _{-\frac{n}{2}} ^{\frac{n}{2}} {P\left(x - y\right)g_{\perp}\left(y\right)dy }\right) \\
&=  f \cdot \left(\int \limits _{-\frac{n}{2}} ^{\frac{n}{2}} {P\left(x - y\right)g_{||}\left(y\right)dy }\right) \\
&=  \left(f_{||} + f_{\perp}\right) \cdot \left(\int \limits _{-\frac{n}{2}} ^{\frac{n}{2}} {P\left(x - y\right)g_{||}\left(y\right)dy }\right) \\
&=  f_{||} \cdot \left(\int \limits _{-\frac{n}{2}} ^{\frac{n}{2}} {P\left(x - y\right)g_{||}\left(y\right)dy }\right) + f_{\perp} \cdot \left(\int \limits _{-\frac{n}{2}} ^{\frac{n}{2}} {P\left(x - y\right)g_{||}\left(y\right)dy }\right) \\
\end{align*}
As $f_{\perp}$ is orthogonal to all polynomials in $x$ of degree not more than $N$, and as $\int \limits _{-\frac{n}{2}} ^{\frac{n}{2}} {P\left(x - y\right)g_{||}\left(y\right)dy }$ is itself a polynomial in $x$ of degree not more than $M$, this means that $f_{\perp} \cdot \left(\int \limits _{-\frac{n}{2}} ^{\frac{n}{2}} {P\left(x - y\right)g_{||}\left(y\right)dy }\right) = 0$. Thus:
\begin{align*}
f \cdot \left(Eg\right) &=  f_{||} \cdot \left(\int \limits _{-\frac{n}{2}} ^{\frac{n}{2}} {P\left(x - y\right)g_{||}\left(y\right)dy }\right) + f_{\perp} \cdot \left(\int \limits _{-\frac{n}{2}} ^{\frac{n}{2}} {P\left(x - y\right)g_{||}\left(y\right)dy }\right) \\
&=  f_{||} \cdot \left(\int \limits _{-\frac{n}{2}} ^{\frac{n}{2}} {P\left(x - y\right)g_{||}\left(y\right)dy }\right) \\
&=  f_{||} \cdot \left(\int \limits _{-\frac{n}{2}} ^{\frac{n}{2}} {R\left(x, y\right)g_{||}\left(y\right)dy }\right) \\
&=  f_{||} \cdot \left(Eg_{||}\right)
\end{align*}
As the polynomials in $x$ (defined on $\left[-\frac{n+a}{2},\frac{n+a}{2}\right]$) with degree not more than $M$ constitute an $\left(M+1\right)$-dimensional space, they have a basis with $M+1$ elements (for example, $\left \lbrace 1, x, x^2, ..., x^M \right \rbrace$). Moreover, applying the Gram-Schmidt construction guarantees the existence of an orthonormal basis of the polynomials of degree not more than $M$, and this basis still has $M+1$ elements. Let $\left \lbrace f_0, f_1, f_2, ..., f_M \right \rbrace$ be this orthonormal basis, or any other orthonormal basis. The polynomials in $y$ (defined on $\left[-\frac{n}{2},\frac{n}{2}\right]$) with degree not more than $M$ have a different orthonormal basis, $\left \lbrace g_0, g_1, g_2, ..., g_M \right \rbrace$. Using these orthonormal bases, $f_{||}$ decomposes as $\sum \limits _{s_x = 0} ^{M} {\left(f_{||} \cdot f_{s_x} \right)f_{s_x}}$ (as can be checked by computing the dot products of each $f_{s_x}$ with both $f_{||}$ and its decomposition). Similarly, $g_{||}$ decomposes as $\sum \limits _{s_y = 0} ^{M} {\left(g_{||} \cdot g_{s_y} \right)g_{s_y}}$. Thus:
\begin{align*}
f \cdot \left(Eg\right) &=  f_{||} \cdot \left(Eg_{||}\right) \\
&= \left(\sum \limits _{s_x = 0} ^{M} {\left(f_{||} \cdot f_{s_x} \right)f_{s_x}}\right) \cdot \left(E\left(\sum \limits _{s_y = 0} ^{M} {\left(g_{||} \cdot g_{s_y} \right)g_{s_y}}\right)\right) \\
&= \left(\sum \limits _{s_x = 0} ^{M} {\left(f_{||} \cdot f_{s_x} \right)f_{s_x}}\right) \cdot \left(\sum \limits _{s_y = 0} ^{M} {\left(g_{||} \cdot g_{s_y} \right)\left(Eg_{s_y}\right)}\right) \\
&= \sum \limits _{s_x = 0} ^{M} {\left(\left(\left(f_{||} \cdot f_{s_x} \right)f_{s_x}\right)\cdot \left(\sum \limits _{s_y = 0} ^{M} {\left(g_{||} \cdot g_{s_y} \right)\left(Eg_{s_y}\right)}\right)\right)} \\
&= \sum \limits _{s_x = 0} ^{M} {\left(\sum \limits _{s_y = 0} ^{M} {\left(g_{||} \cdot g_{s_y} \right)\left(\left(\left(f_{||} \cdot f_{s_x} \right)f_{s_x}\right) \cdot \left(Eg_{s_y}\right)\right)}\right)} \\
&= \sum \limits _{s_x = 0} ^{M} {\sum \limits _{s_y = 0} ^{M} {\left(\left(f_{||} \cdot f_{s_x} \right)\left(g_{||} \cdot g_{s_y} \right)\left(\left(f_{s_x}\right) \cdot \left(Eg_{s_y}\right)\right)\right)}} \\
\end{align*}
As $f_0, f_1, f_2, \cdots, f_M$ are all polynomials with degree at most $M$, and as $f_{\perp}$ is orthogonal to all polynomials with degree at most $M$, it follows that $f_{\perp} \cdot f_{s_x} = 0$. Thus,
\begin{align*}
f \cdot f_{s_x} &= \left(f_{||} + f_{\perp}\right) \cdot f_{s_x} \\
&= \left(f_{||} \cdot f_{s_x} \right) + \left(f_{\perp} \cdot f_{s_x} \right) \\
&= \left(f_{||} \cdot f_{s_x} \right) + 0 \\
&= f_{||} \cdot f_{s_x}
\end{align*}
For the same reason, $g \cdot g_{s_y} = g_{||} \cdot g_{s_y}$. Thus:
\begin{align*}
f \cdot \left(Eg\right) &= \sum \limits _{s_x = 0} ^{M} {\sum \limits _{s_y = 0} ^{M} {\left(\left(f_{||} \cdot f_{s_x} \right)\left(g_{||} \cdot g_{s_y} \right)\left(\left(f_{s_x}\right) \cdot \left(Eg_{s_y}\right)\right)\right)}} \\
&= \sum \limits _{s_x = 0} ^{M} {\sum \limits _{s_y = 0} ^{M} {\left(\left(f \cdot f_{s_x} \right)\left(g \cdot g_{s_y} \right)\left(\left(f_{s_x}\right) \cdot \left(Eg_{s_y}\right)\right)\right)}} \\
&=
\left( 
\begin{matrix}
f \cdot f_0 & \cdots & f \cdot f_M \\
\end{matrix}
\right)
\left( 
\begin{matrix}
f_0 \cdot \left(Eg_0\right) & \cdots & f_0 \cdot \left(Eg_M\right) \\
\cdots & \cdots & \cdots \\
f_M \cdot \left(Eg_0\right) & \cdots & f_M \cdot \left(Eg_M\right) \\
\end{matrix}
\right)
\left( 
\begin{matrix}
g \cdot g_{0} \\
\cdots \\
g \cdot g_M \\
\end{matrix}
\right)
\end{align*}
This is the equality claimed in this lemma. Also, indeed, $f_0, f_1, \cdots, f_M$ are orthonormal and are all polynomials with degree not greater than $M$ (as they were defined to be this way), and where $g_0, g_1, \cdots, g_M$ are orthonormal and are all polynomials with degree not greater than $M$ (again, as they were defined to be this way).
\end{proof}
\begin{lem}
\label{convexminusonelemma}
Let $v_1$, ..., $v_{N+2}$ be points in $\Re^N$. Then, if $v$ is a convex combination of these points, $v$ can also be expressed as a convex combination of at most $N+1$ of them, not necessarily the same ones for every choice of $v$.
\end{lem}
\emph{Proof sketch}. Consider $v$ as a convex combination of exactly $N+2$ points, $v_1$, $v_2$, ..., and $v_{N+2}$. The differences between the $N+2$ points span a space of no more than $N$ dimensions, so $N$ differences suffice to span this space. That allows one of $v_2 - v_1$, $v_3 - v_1$, ...,and $v_{N+2} - v_1$ to be expressed in terms of the rest, or, equivalently, allows $0$ to be a linear combination of these differences. Add this alternative way of representing $0$ to $v$, which leaves the sum of the coefficients untouched (because each difference has $0$ as its sum of coefficients). In fact, the same is true for adding any multiple of this representation of $0$ to $v$. Some multiples will keep all the coefficients nonnegative; some will not. Small multiples will keep all the coefficients nonnegative (unless one was zero to begin with, but then that term can just be removed). Find the edges of region (of multiples of the zero) that keep all the coefficients nonnegative. On either edge of the region, one or more coefficients is zero, and the rest are nonnegative. Now remove whichever terms have the zero coefficient, and the result is a convex combination of $N+1$ points or fewer.
\begin{proof}
Let $v = \sum \limits _{s = 1} ^{N+2} {c_s v_s}$, where $\sum \limits _{s = 1} ^{N+2} {c_s} = 1$ and $0 \leq c_s \leq 1$ for all $s \in \left \lbrace 1, 2, ..., N+2 \right \rbrace$ (in other words, $v$ is a convex combination of the $v_s$). If $v_1$ is not the zero vector, then subtract $v_1$ from all the $v_s$ (including $v_1$ itself) and also from $v$, and if $v$ was a convex combination of the $v_s$, then it is still a convex combination of the $v_s$, with the same choices of $c_s$; conversely, if $v$ is now a convex combination of the $v_s$, it was a convex combination of the $v_s$, again with the same choices of $c_s$.

Let $S$ be the space spanned by the $v_s$; this may or may not be $\Re^N$. Regardless, it is possible to select $N$ vectors from the $v_s$ that span $S$. Add $v_1$ to this set (or if $v_1$ is already in, add another of the $v_s$ instead). Let $v_{N+2}$ be the one of the $v_s$ that is not in this set; if it is not $v_{N+2}$, then reassign the labels of the $v_s$ so that it is indeed $v_{N+2}$ outside the set (and so that $v_1$ remains 0).

Given that $v_1$, ..., $v_{N+1}$ span $S$, it follows that $v_{N+2} = \sum \limits _{s = 1} ^{N+1} {d_s v_s}$ for some real numbers $d_s$. Equivalently, if $d_{N+2}$ is set to be $-1$, $\sum \limits _{s = 1} ^{N+2} {d_s v_s} = 0$. As $v_1 = 0$, it can be further specified that $d_1 = -\sum \limits _{s = 2} ^{N+2} {d_s}$, so that the sum of all the $d_s$ is zero.

Thus, not only does $v = \sum \limits _{s = 1} ^{N+2} {c_s v_s}$, but also, $v = \sum \limits _{s = 1} ^{N+2} {c_s v_s} + b\left(\sum \limits _{s = 1} ^{N+2} {d_s v_s}\right) = v$ for any real number $b$. Simplifying the right side of this equation yields $v = \sum \limits _{s = 1} ^{N+2} {\left(c_s + b d_s\right) v_s}$, with $\sum \limits _{s = 1} ^{N+2} {\left(c_s + b d_s\right)}= \sum \limits _{s = 1} ^{N+2} {c_s}+b\sum \limits _{s = 1} ^{N+2} {d_s} = 1 + b(0) = 1$, so, provided that all coefficients of $\sum \limits _{s = 1} ^{N+2} {\left(c_s + b d_s\right) v_s}$ are nonnegative, this expression gives a different way to write $v$ as a convex combination of $v_1$, ..., $v_{N+2}$.

All that remains is to select a value of $b$ that leaves all coefficients nonnegative, but that sets at least one of them to zero. There is at least one negative $d_s$ (for instance, $d_{N+2}$). For all $s$ where $d_s$ is negative, let $t$ be the one for which $\frac{-c_s}{d_s}$ is lowest, and set $b = \frac{-c_t}{d_t}$. $b$ is a nonnegative number, because $c_t \geq 0$. Thus, for any $s$ where $d_s \geq 0$, $c_s + b d_s \geq 0$ is satisfied automatically. If $d_s < 0$, then $c_s + b d_s = c_s + \frac{-c_t}{d_t}d_s \geq c_s + \frac{-c_s}{d_s}d_s = 0$. Thus, all coefficients are nonnegative. The coefficient with index $t$ is zero, because $c_t + b d_t = c_t + \frac{-c_t}{d_t}d_t = 0$. Thus, this is a value of $b$ that makes $\sum \limits _{s = 1} ^{N+2} {\left(c_s + b d_s\right) v_s}$ a representation of $v$ as a convex combination of all the $v_1$, ..., $v_{N+2}$, with at least one coefficient equal to $0$. Remove this term, and the result is a convex combination of at most $N+1$ of $v_1$, ..., $v_{N+2}$.
\end{proof}
\begin{lem}
\label{convexfinitelemma}
Let $v_1$, ..., $v_{M}$ be points in $\Re^N$, with $M > N$. Then, if $v$ is a convex combination of these points, $v$ can also be expressed as a convex combination of at most $N+1$ of them, not necessarily the same ones for every choice of $v$.
\end{lem}
\emph{Proof sketch}. For a convex combination of more than $N+2$ points, consider just the first $N+2$ points, with rescaled coefficients, so that the rescaled coefficients sum to $1$. That is a convex combination of $N+2$ points, so it can be rewritten as a convex combination of no more than $N+1$ of those $N+2$ points by the previous lemma. Scale the $N+1$ coefficients back, so that their sum is the same as was the sum of the $N+2$ original coefficients, and replace the part of the convex combination corresponding to the first $N+2$ points with this new expression. That is still a representation of $v$ as a convex combination of points, but this time, it is a convex combination of only $M-1$ points. As long as there are $N+2$ or more points, this procedure can remove one, so repeat this procedure until there are only $N+1$ or fewer points.
\begin{proof}
Let $v = \sum \limits _{s = 1} ^{M} {c_s v_s}$, where $\sum \limits _{s = 1} ^{M} {c_s} = 1$ and $0 \leq c_s \leq 1$ for all $s \in \left \lbrace 1, 2, ..., M \right \rbrace$ (in other words, $v$ is a convex combination of the $v_s$). Also, let at least one of $c_1$, $c_2$, ..., $c_{N+2}$ be nonzero. (Reshuffle the points if necessary to make it so. There is at least one point with a nonzero coefficient.) Now:
\begin{align*}
v &= \sum \limits _{s = 1} ^{M} {c_s v_s} \\
&= \left( \sum \limits _{s = 1} ^{N+2} {c_s v_s} \right) + \left( \sum \limits _{s = N+3} ^{M} {c_s v_s} \right) \\
&= \left(\frac{\sum \limits _{u = 1} ^{N+2} {c_u}}{\sum \limits _{u = 1} ^{N+2} {c_u}}\right)\left( \sum \limits _{s = 1} ^{N+2} {c_s v_s} \right) + \left( \sum \limits _{s = N+3} ^{M} {c_s v_s} \right) \\
&= \left(\sum \limits _{u = 1} ^{N+2} {c_u}\right)\left( \sum \limits _{s = 1} ^{N+2} {\left(\frac{c_s}{\sum \limits _{u = 1} ^{N+2} {c_u}}\right) v_s} \right) + \left( \sum \limits _{s = N+3} ^{M} {c_s v_s} \right)
\end{align*}
As $\sum \limits _{s = 1} ^{N+2} {\left(\frac{c_s}{\sum \limits _{u = 1} ^{N+2} {c_u}}\right) v_s}$ is a convex combination of $v_1$, ..., $v_{N+2}$ (as its coefficients are nonnegative and sum to 1), by Lemma \ref{convexminusonelemma}, $\sum \limits _{s = 1} ^{N+2} {\left(\frac{c_s}{\sum \limits _{u = 1} ^{N+2} {c_u}}\right) v_s}$ is a convex combination of at most $N+1$ vectors chosen from $v_1$, ..., $v_{N+2}$. Let $w_1$, ..., $w_{N+1}$ be these $N+1$ vectors. Thus, for some nonnegative numbers $d_t$ summing to 1, $\sum \limits _{s = 1} ^{N+2} {\left(\frac{c_s}{\sum \limits _{u = 1} ^{N+2} {c_u}}\right) v_s} = \sum \limits _{t = 1} ^{N+1} {d_t w_t}$. Therefore:
\begin{align*}
& \left(\sum \limits _{u = 1} ^{N+2} {c_u}\right)\left( \sum \limits _{s = 1} ^{N+2} {\left(\frac{c_s}{\sum \limits _{u = 1} ^{N+2} {c_u}}\right) v_s} \right) + \left( \sum \limits _{s = N+3} ^{M} {c_s v_s} \right) \\
&= \left(\sum \limits _{u = 1} ^{N+2} {c_u}\right)\left(\sum \limits _{t = 1} ^{N+1} {d_t w_t} \right) + \left( \sum \limits _{s = N+3} ^{M} {c_s v_s} \right) \\
&= \left(\sum \limits _{t = 1} ^{N+1} {\left(d_t\sum \limits _{u = 1} ^{N+2} {c_u}\right) w_t} \right) + \left( \sum \limits _{s = N+3} ^{M} {c_s v_s} \right) \\
\end{align*}
All the coefficients sum to $1$, as in the first sum, the coefficients sum to $\sum \limits _{t = 1} ^{N+1} {\left(d_t\sum \limits _{u = 1} ^{N+2} {c_u}\right)} = \sum \limits _{u = 1} ^{N+2} {c_u}\sum \limits _{t = 1} ^{N+1} {d_t} = \sum \limits _{u = 1} ^{N+2} {c_u}= \sum \limits _{s = 1} ^{N+2} {c_s}$, and in the second term, the coefficients sum to $\sum \limits _{s = N+3} ^{M} {c_s v_s}$, for an overall sum of $\sum \limits _{s = 1} ^{M} {c_s v_s}$, which is $1$. Also, all the coefficients are nonnegative (because $\sum \limits _{s = 1} ^{N+2} {c_s}$ is nonnegative). Thus, this is a convex combination of $w_1$, ..., $w_{N+1}$, $v_{N+3}$, ..., $v_M$, and this list includes only $M - 1$ vectors.
This procedure can be repeated as many times as needed, to bring the number of vectors down to $N+2$ (or fewer, in which case the lemma statement is true); at that point, Lemma \ref{convexminusonelemma} (or yet another run of this procedure, with the sum from $N+3$ to $M$ having no terms and therefore being zero) can be applied for a final reduction to at most $N+1$ points.
\end{proof}
\begin{lem}
\label{convexlemma}
Let $h$ be a continuous function from $\left[-\frac{n}{2}, \frac{n}{2}\right]$ to $\Re^N$. Then, if $v = \int \limits _{\frac{n}{2}} ^{\frac{n}{2}} {f\left(x\right)h\left(x\right) dx}$ with $\int \limits _{\frac{n}{2}} ^{\frac{n}{2}} {f\left(x\right) dx} = 1$ and $f\left(x\right) \geq 0$ for all $x$ in $\left[-\frac{n}{2}, \frac{n}{2}\right]$ (in other words, $v$ is in the convex hull of the path $h$), then $v$ is a convex combination of at most $N+1$ values of the range of $h$ (in other words, of at most $N+1$ points on the path $h$).
\end{lem}
\emph{Proof sketch}. From the previous lemma, it is possible to reduce a convex combination of finitely many points in an $N$-dimensional space to a convex combination of only $N+1$ of those points. For a point $v$ that is the convex combination of infinitely-many points on the path, divide the integral into $M$ pieces. In each piece, keep $f\left(x\right)$ exactly as it was, but approximate $h\left(x\right)$ by $h$ at the left (lowest $x$) endpoint. As $M$ gets big, these approximations approach $v$. Each of the approximations is a convex combination of finitely many points, so the previous lemma shows that each approximation is a convex combination of only $N+1$ points. Write it this way for each approximation. Then, at least some of the approximations form a convergent sequence, in that the first points of the convex combinations converge, as do the second points, and so on until the $\left(N+1\right)$st points, and so do the corresponding coefficients. For each point or coefficient, take the limit of the corresponding point or coefficient in the convergent subsequence, and this serves as a representation of $v$ as a convex combination of $N+1$ points.

\begin{proof}
Divide $\left[-\frac{n}{2}, \frac{n}{2}\right]$ into $M$ equal-size intervals $\left[-\frac{n}{2}, -\frac{n}{2} + \frac{n}{M}\right]$, $\left[-\frac{n}{2} + \frac{n}{M}, -\frac{n}{2} + 2\frac{n}{M}\right]$, ..., $\left[-\frac{n}{2} + \left(M-1\right)\frac{n}{M}, \frac{n}{2}\right]$. Then,
\begin{align*}
v &= \int \limits _{\frac{n}{2}} ^{\frac{n}{2}} {f\left(x\right)h\left(x\right) dx} \\
&= \sum \limits _{s = 0} ^{M-1} {\int \limits _{-\frac{n}{2} + s\frac{n}{M}} ^{-\frac{n}{2} + \left(s+1\right)\frac{n}{M}} {f\left(x\right)h\left(x\right) dx}} \\
&= \sum \limits _{s = 0} ^{M-1} {\int \limits _{-\frac{n}{2} + s\frac{n}{M}} ^{-\frac{n}{2} + \left(s+1\right)\frac{n}{M}} {f\left(x\right)
\left(
\begin{matrix}
h_1\left(x\right) \\
... \\
h_N\left(x\right)
\end{matrix}
\right)
dx}} \\
& \text{(with $h_1\left(x\right)$, ..., $h_N\left(x\right)$ being the components of $h\left(x\right)$)} \\
&= \sum \limits _{s = 0} ^{M-1} {\int \limits _{-\frac{n}{2} + s\frac{n}{M}} ^{-\frac{n}{2} + \left(s+1\right)\frac{n}{M}} {
\left(
\begin{matrix}
f\left(x\right)h_1\left(x\right) \\
... \\
f\left(x\right)h_N\left(x\right)
\end{matrix}
\right)
dx}} \\
&= \sum \limits _{s = 0} ^{M-1} {
\left(
\begin{matrix}
\int \limits _{-\frac{n}{2} + s\frac{n}{M}} ^{-\frac{n}{2} + \left(s+1\right)\frac{n}{M}}{f\left(x\right)h_1\left(x\right) dx} \\
... \\
\int \limits _{-\frac{n}{2} + s\frac{n}{M}} ^{-\frac{n}{2} + \left(s+1\right)\frac{n}{M}}{f\left(x\right)h_N\left(x\right) dx}
\end{matrix}
\right)
} \\
&= 
\left(
\begin{matrix}
\sum \limits _{s = 0} ^{M-1} {\int \limits _{-\frac{n}{2} + s\frac{n}{M}} ^{-\frac{n}{2} + \left(s+1\right)\frac{n}{M}}{f\left(x\right)h_1\left(x\right) dx}} \\
... \\
\sum \limits _{s = 0} ^{M-1} {\int \limits _{-\frac{n}{2} + s\frac{n}{M}} ^{-\frac{n}{2} + \left(s+1\right)\frac{n}{M}}{f\left(x\right)h_N\left(x\right) dx}}
\end{matrix}
\right)
 \\
\end{align*}
Let $v_t = \sum \limits _{s = 0} ^{M-1} {\int \limits _{-\frac{n}{2} + s\frac{n}{M}} ^{-\frac{n}{2} + \left(s+1\right)\frac{n}{M}}{f\left(x\right)h_t\left(x\right) dx}}$. In other words, $v_t$ is the $t$-th component of $v$. Then, because $f\left(x\right)$ is nonnegative:
\begin{align*}
v_t &= \sum \limits _{s = 0} ^{M-1} {\int \limits _{-\frac{n}{2} + s\frac{n}{M}} ^{-\frac{n}{2} + \left(s+1\right)\frac{n}{M}}{f\left(x\right)h_t\left(x\right) dx}} \\
&\leq \sum \limits _{s = 0} ^{M-1} {\int \limits _{-\frac{n}{2} + s\frac{n}{M}} ^{-\frac{n}{2} + \left(s+1\right)\frac{n}{M}}{f\left(x\right)\left(\max \limits _{-\frac{n}{2} + s\frac{n}{M} \leq b \leq-\frac{n}{2} + \left(s+1\right)\frac{n}{M}} {h_t\left(b\right)}\right) dx}} \\
\end{align*}
In other words, if the $h_t\left(x\right)$ in each interval get increased to the maximum $h_t\left(x\right)$  in that interval (which exists because $h_t$ is continuous), then $v_t$ does not decrease. A similar statement gives a lower bound to $v_t$:
\begin{align*}
v_t &= \sum \limits _{s = 0} ^{M-1} {\int \limits _{-\frac{n}{2} + s\frac{n}{M}} ^{-\frac{n}{2} + \left(s+1\right)\frac{n}{M}}{f\left(x\right)h_t\left(x\right) dx}} \\
&\geq \sum \limits _{s = 0} ^{M-1} {\int \limits _{-\frac{n}{2} + s\frac{n}{M}} ^{-\frac{n}{2} + \left(s+1\right)\frac{n}{M}}{f\left(x\right)\left(\min \limits _{-\frac{n}{2} + s\frac{n}{M} \leq b \leq-\frac{n}{2} + \left(s+1\right)\frac{n}{M}} {h_t\left(b\right)}\right) dx}} \\
\end{align*}
Let $b_{s, t, max}$ satisfy $-\frac{n}{2} + s\frac{n}{M} \leq b \leq-\frac{n}{2} + \left(s+1\right)\frac{n}{M}$, and be such that $h_t\left(b_{s, t, max}\right) = \max \limits _{-\frac{n}{2} + s\frac{n}{M} \leq b \leq-\frac{n}{2} + \left(s+1\right)\frac{n}{M}} {\left(h_t\left(b\right)\right)}$. In other words, in the maximum expression, $b_{s, t, max}$ is the $b$-value achieving that maximum, which exists because $h_t$ is continuous and the maximum is taken over a closed interval. Similarly, let $b_{s, t, min}$ satisfy $-\frac{n}{2} + s\frac{n}{M} \leq b \leq-\frac{n}{2} + \left(s+1\right)\frac{n}{M}$, and be such that $h_t\left(b_{s, t, min}\right) = \min \limits _{-\frac{n}{2} + s\frac{n}{M} \leq b \leq-\frac{n}{2} + \left(s+1\right)\frac{n}{M}} {\left(h_t\left(b\right)\right)}$. Thus:
\[
v_t \leq \sum \limits _{s = 0} ^{M-1} {\int \limits _{-\frac{n}{2} + s\frac{n}{M}} ^{-\frac{n}{2} + \left(s+1\right)\frac{n}{M}}{f\left(x\right) h_t\left(b_{s, t, max}\right) dx}}
\]
and
\[
v_t \geq \sum \limits _{s = 0} ^{M-1} {\int \limits _{-\frac{n}{2} + s\frac{n}{M}} ^{-\frac{n}{2} + \left(s+1\right)\frac{n}{M}}{f\left(x\right) h_t\left(b_{s, t, min}\right) dx}}
\]
The difference between the bounds is:
\begin{align*}
&\sum \limits _{s = 0} ^{M-1} {\int \limits _{-\frac{n}{2} + s\frac{n}{M}} ^{-\frac{n}{2} + \left(s+1\right)\frac{n}{M}}{f\left(x\right) h_t\left(b_{s, t, max}\right) dx}} - \sum \limits _{s = 0} ^{M-1} {\int \limits _{-\frac{n}{2} + s\frac{n}{M}} ^{-\frac{n}{2} + \left(s+1\right)\frac{n}{M}}{f\left(x\right) h_t\left(b_{s, t, min}\right) dx}} \\
&= \sum \limits _{s = 0} ^{M-1} {\left(\int \limits _{-\frac{n}{2} + s\frac{n}{M}} ^{-\frac{n}{2} + \left(s+1\right)\frac{n}{M}}{f\left(x\right) h_t\left(b_{s, t, max}\right) dx}-\int \limits _{-\frac{n}{2} + s\frac{n}{M}} ^{-\frac{n}{2} + \left(s+1\right)\frac{n}{M}}{f\left(x\right) h_t\left(b_{s, t, min}\right) dx}\right)} \\
&= \sum \limits _{s = 0} ^{M-1} {\int \limits _{-\frac{n}{2} + s\frac{n}{M}} ^{-\frac{n}{2} + \left(s+1\right)\frac{n}{M}}{\left(f\left(x\right) h_t\left(b_{s, t, max}\right)-f\left(x\right) h_t\left(b_{s, t, min}\right)\right) dx}} \\
&= \sum \limits _{s = 0} ^{M-1} {\int \limits _{-\frac{n}{2} + s\frac{n}{M}} ^{-\frac{n}{2} + \left(s+1\right)\frac{n}{M}}{f\left(x\right)\left( h_t\left(b_{s, t, max}\right)- h_t\left(b_{s, t, min}\right)\right) dx}} \\
\end{align*}
Since $h_t$ is a continuous function on a closed interval (which is $\left[-\frac{n}{2},\frac{n}{2}\right]$), it is uniformly continuous. Therefore, for any $\epsilon > 0$, there exists a $\delta > 0$, such that for all $x_1$ and $x_2$ in $\left[-\frac{n}{2},\frac{n}{2}\right]$, if $\left|x_1-x_2\right| < \delta$, then $\left|h_t\left(x_1\right)-h_t\left(x_2\right)\right| < \epsilon$. For each $\epsilon$, select a $\delta$ where this is so, and choose any $M > \frac{n}{\delta}$. For any such $M$, $-\frac{n}{2} + s\frac{n}{M}$ and $-\frac{n}{2} + \left(s+1\right)\frac{n}{M}$ are only $\frac{n}{M}$ apart, which is less than $\frac{n}{\left(\frac{n}{\delta}\right)}$, or, equivalently, less than $\delta$. That means that $b_{s, t, min}$ and $b_{s, t, max}$ are less than $\delta$ apart, as they are both in $\left[-\frac{n}{2} + s\frac{n}{M}, -\frac{n}{2} + \left(s+1\right)\frac{n}{M}\right]$ and the endpoints of this interval are less than $\delta$ apart. Thus, $\left|b_{s, t, min} - b_{s, t, max}\right| < \delta$, and thus, $\left|h_t\left(b_{s, t, min}\right) - h_t\left(b_{s, t, max}\right)\right| < \epsilon$. Thus, for any $\epsilon$, it is possible, just by setting $M$ to be big enough, to set:
\begin{align*}
& \left| \sum \limits _{s = 0} ^{M-1} {\int \limits _{-\frac{n}{2} + s\frac{n}{M}} ^{-\frac{n}{2} + \left(s+1\right)\frac{n}{M}}{f\left(x\right)\left( h_t\left(b_{s, t, max}\right)- h_t\left(b_{s, t, min}\right)\right) dx}} \right| \\
& \leq  \sum \limits _{s = 0} ^{M-1} {\left|\int \limits _{-\frac{n}{2} + s\frac{n}{M}} ^{-\frac{n}{2} + \left(s+1\right)\frac{n}{M}}{f\left(x\right)\left( h_t\left(b_{s, t, max}\right)- h_t\left(b_{s, t, min}\right)\right) dx}\right|} \\
& \leq  \sum \limits _{s = 0} ^{M-1} {\int \limits _{-\frac{n}{2} + s\frac{n}{M}} ^{-\frac{n}{2} + \left(s+1\right)\frac{n}{M}}{\left|f
\left(x\right)\left( h_t\left(b_{s, t, max}\right)- h_t\left(b_{s, t, min}\right)\right) \right|dx}} \\
& \leq  \sum \limits _{s = 0} ^{M-1} {\int \limits _{-\frac{n}{2} + s\frac{n}{M}} ^{-\frac{n}{2} + \left(s+1\right)\frac{n}{M}}{f\left(x\right)\left| h_t\left(b_{s, t, max}\right)- h_t\left(b_{s, t, min}\right)\right|dx}} \\
& \leq  \sum \limits _{s = 0} ^{M-1} {\int \limits _{-\frac{n}{2} + s\frac{n}{M}} ^{-\frac{n}{2} + \left(s+1\right)\frac{n}{M}}{f\left(x\right)\left(\epsilon\right)dx}} \\
& \leq \epsilon \sum \limits _{s = 0} ^{M-1} {\int \limits _{-\frac{n}{2} + s\frac{n}{M}} ^{-\frac{n}{2} + \left(s+1\right)\frac{n}{M}}{f\left(x\right)dx}} \\
& \leq \epsilon \int \limits _{-\frac{n}{2}} ^{\frac{n}{2}}{f\left(x\right)dx} \\
& \leq \epsilon \\
\end{align*}
Thus,
\[
\lim \limits _{M \rightarrow \infty} \sum \limits _{s = 0} ^{M-1} {\int \limits _{-\frac{n}{2} + s\frac{n}{M}} ^{-\frac{n}{2} + \left(s+1\right)\frac{n}{M}}{f\left(x\right)\left( h_t\left(b_{s, t, max}\right)- h_t\left(b_{s, t, min}\right)\right) dx}} = 0
\]
which means that the bounds on $v_t$ approach each other. Therefore, both bounds, as well as anything always at or between them, must converge to $v_t$.

One thing guaranteed to be at or between the bounds is the left sum for 
\[
v_t = \sum \limits _{s = 0} ^{M-1} {\int \limits _{-\frac{n}{2} + s\frac{n}{M}} ^{-\frac{n}{2} + \left(s+1\right)\frac{n}{M}}{f\left(x\right)h_t\left(x\right) dx}}
\] which is
\[
\sum \limits _{s = 0} ^{M-1} {\left(\left(h_t\left(-\frac{n}{2} + s\frac{n}{M}\right)\right)\int \limits _{-\frac{n}{2} + s\frac{n}{M}} ^{-\frac{n}{2} + \left(s+1\right)\frac{n}{M}}{f\left(x\right) dx}\right)}
\]
Recombining these left sums gives:
\begin{align*}
&
\left(
\begin{matrix}
\sum \limits _{s = 0} ^{M-1} {\left(\left(h_1\left(-\frac{n}{2} + s\frac{n}{M}\right)\right)\int \limits _{-\frac{n}{2} + s\frac{n}{M}} ^{-\frac{n}{2} + \left(s+1\right)\frac{n}{M}}{f\left(x\right) dx}\right)} \\
... \\
\sum \limits _{s = 0} ^{M-1} {\left(\left(h_N\left(-\frac{n}{2} + s\frac{n}{M}\right)\right)\int \limits _{-\frac{n}{2} + s\frac{n}{M}} ^{-\frac{n}{2} + \left(s+1\right)\frac{n}{M}}{f\left(x\right) dx}\right)} \\
\end{matrix}
\right) \\
& = \sum \limits _{s = 0} ^{M-1}
{
\left(
\begin{matrix}
\left(h_1\left(-\frac{n}{2} + s\frac{n}{M}\right)\right)\int \limits _{-\frac{n}{2} + s\frac{n}{M}} ^{-\frac{n}{2} + \left(s+1\right)\frac{n}{M}}{f\left(x\right) dx} \\
... \\
\left(h_N\left(-\frac{n}{2} + s\frac{n}{M}\right)\right)\int \limits _{-\frac{n}{2} + s\frac{n}{M}} ^{-\frac{n}{2} + \left(s+1\right)\frac{n}{M}}{f\left(x\right) dx} \\
\end{matrix}
\right)
} \\
& = \sum \limits _{s = 0} ^{M-1}
{
\left(
\int \limits _{-\frac{n}{2} + s\frac{n}{M}} ^{-\frac{n}{2} + \left(s+1\right)\frac{n}{M}}{f\left(x\right) dx}
\left(
\begin{matrix}
h_1\left(-\frac{n}{2} + s\frac{n}{M}\right) \\
... \\
h_N\left(-\frac{n}{2} + s\frac{n}{M}\right) \\
\end{matrix}
\right)
\right)
} \\
& = \sum \limits _{s = 0} ^{M-1}
{
\left(
\left(\int \limits _{-\frac{n}{2} + s\frac{n}{M}} ^{-\frac{n}{2} + \left(s+1\right)\frac{n}{M}}{f\left(x\right) dx}\right) h\left(-\frac{n}{2} + s\frac{n}{M}\right) \right)}
\end{align*}
This is actually a convex combination of the $h\left(-\frac{n}{2} + s\frac{n}{M}\right)$, as the coefficients of the $h\left(-\frac{n}{2} + s\frac{n}{M}\right)$ are $\int \limits _{-\frac{n}{2} + s\frac{n}{M}} ^{-\frac{n}{2} + \left(s+1\right)\frac{n}{M}}{f\left(x\right) dx}$, which are all nonnegative, and which sum to:
\begin{align*}
& \sum _{s = 0} ^{M-1} {\int \limits _{-\frac{n}{2} + s\frac{n}{M}} ^{-\frac{n}{2} + \left(s+1\right)\frac{n}{M}}{f\left(x\right) dx}} \\
& = \int \limits _{-\frac{n}{2}} ^{\frac{n}{2}}{f\left(x\right) dx} \\
& = 1
\end{align*}
As such, by Lemma \ref{convexfinitelemma}, the left sum for $v$ (which is 
\[
\sum \limits _{s = 0} ^{M-1}{\left(\left(\int \limits _{-\frac{n}{2} + s\frac{n}{M}} ^{-\frac{n}{2} + \left(s+1\right)\frac{n}{M}}{f\left(x\right) dx}\right) h\left(-\frac{n}{2} + s\frac{n}{M}\right) \right)}
\]
) is a convex combination of at most $N+1$ of the $h\left(-\frac{n}{2} + s\frac{n}{M}\right)$. If it has fewer than $N+1$ points, add more points to it with coefficient $0$ to make it a convex combination of exactly $N+1$ points.Equivalently, for some nonnegative $c_{q, M}$ (where $q$ goes from $1$ to $N+1$) that sum to $1$, and where the $x_{q, M}$ are selected from the $-\frac{n}{2}+s\frac{n}{M}$ (for the same $M$),
\[
\sum \limits _{s = 0} ^{M-1}{\left(\left(\int \limits _{-\frac{n}{2} + s\frac{n}{M}} ^{-\frac{n}{2} + \left(s+1\right)\frac{n}{M}}{f\left(x\right) dx}\right) h\left(-\frac{n}{2} + s\frac{n}{M}\right) \right)} = \sum \limits _{q = 1} ^{N+1}{\left(c_{q, M} h\left(x_{q, M}\right) \right)}
\] 
Also, as $M$ gets large, the left sum for $v$ approaches $v$, because each of its components was already shown to approach the corresponding component of $v$. Therefore,
\[
\lim \limits _{M \rightarrow \infty} {\sum \limits _{q = 1} ^{N+1}{\left(c_{q, M} h\left(x_{q, M}\right) \right)}} = v
\]
The sequence of vectors
\[
\left \lbrace
\left(
\begin{matrix}
c_{1, M} \\
... \\
c_{N+1, M} \\
x_{1, M} \\
... \\
x_{N+1, M}
\end{matrix}
\right)
\right \rbrace
\]
has a convergent (in all components) subsequence. The reason for that is that there is a subsequence where the first components converge (the first components being bounded), and from there, there is a sub-subsequence where the second components converge, and this process can be repeated for each component. Let $\left \lbrace M_p \right \rbrace _{p = 1} ^{\infty}$ be the choices of $M$ in this convergent subsequence. That means that all the $\lim \limits _{p \rightarrow \infty} c_{q, M_p}$ and all the $\lim \limits _{p \rightarrow \infty} x_{q, M_p}$ exist. As such, the expression,
\[
\sum \limits _{q = 1} ^{N+1} {\left(\lim \limits _{p \rightarrow \infty} {c_{q, M_p}}
\left(
\begin{matrix}
h_1\left(\lim \limits _{p \rightarrow \infty} {x_{q, M_p}}\right) \\
... \\
h_N\left(\lim \limits _{p \rightarrow \infty} {x_{q, M_p}}\right)
\end{matrix}
\right)
\right)}
\]
has no non-existent limits in it, so it can be simplified using the usual rules for limits to
\[
\lim \limits _{p \rightarrow \infty} {
\sum \limits _{q = 1} ^{N+1} {\left(c_{q, M_p}
\left(
\begin{matrix}
h_1\left(x_{q, M_p}\right)\\
... \\
h_N\left(x_{q, M_p}\right)
\end{matrix}
\right)
\right)}}
\]
which is
\[
\lim \limits _{p \rightarrow \infty} {\sum \limits _{q = 1} ^{N+1} {\left(c_{q, M_p}h\left(x_{q, M_p}\right)\right)}}
\]
Given that it is already known that
\[
\lim \limits _{M \rightarrow \infty} {\sum \limits _{q = 1} ^{N+1}{\left(c_{q, M} h\left(x_{q, M}\right) \right)}} = v
\]
it follows that
\[
\lim \limits _{p \rightarrow \infty} {\sum \limits _{q = 1} ^{N+1} {\left(c_{q, M_p}h\left(x_{q, M_p}\right)\right)}} = v
\]
This means that
\[
v =
\sum \limits _{q = 1} ^{N+1} {\left(\lim \limits _{p \rightarrow \infty} {c_{q, M_p}}
\left(
\begin{matrix}
h_1\left(\lim \limits _{p \rightarrow \infty} {x_{q, M_p}}\right) \\
... \\
h_N\left(\lim \limits _{p \rightarrow \infty} {x_{q, M_p}}\right)
\end{matrix}
\right)
\right)}
\]
As long as each (as $q$ varies)
\[
\left(
\begin{matrix}
h_1\left(\lim \limits _{p \rightarrow \infty} {x_{q, M_p}}\right) \\
... \\
h_N\left(\lim \limits _{p \rightarrow \infty} {x_{q, M_p}}\right) 
\end{matrix}
\right)
\]
is a point on the path $h$, this is a representation of $v$ as a convex combination of at most $N+1$ points on the path $h$, as the $c_{q, M_p}$ sum to 1, and thus so do their limits as $p$ approaches infinity, and as the $c_{q, M_p}$ are all nonnegative, and thus so are their limits as $p$ approaches infinity.

Indeed, since all the $h_t$ are continuous, it follows that
\[
\left(
\begin{matrix}
h_1\left(\lim \limits _{p \rightarrow \infty} {x_{q, M_p}}\right) \\
... \\
h_N\left(\lim \limits _{p \rightarrow \infty} {x_{q, M_p}}\right)
\end{matrix}
\right) = h\left(\lim \limits _{p \rightarrow \infty} {\left(x_{q, M_p}\right)}\right)
\]
The limit $\lim \limits _{p \rightarrow \infty} {\left(x_{q, M_p}\right)}$ exists (since the $M_p$ were selected to make it exist) and is in $\left[-\frac{n}{2}, \frac{n}{2}\right]$ (since all the $\left(x_{q, M_p}\right)$ are in that interval), so 
$
\left(
\begin{matrix}
h_1\left(\lim \limits _{p \rightarrow \infty} {x_{q, M_p}}\right) \\
... \\
h_N\left(\lim \limits _{p \rightarrow \infty} {x_{q, M_p}}\right)
\end{matrix}
\right)
$
is indeed on the path $h$. Thus, $v$ indeed is a convex combination of at most $N+1$ points on the path $h$.

\end{proof}

\begin{lem}
\label{neglemma}
For any function $h$, let $h_{-}$ be the function defined as $h_{-}\left(z\right) = h\left(-z\right)$. Then, $\left(Eg\right)_{-} = E\left(g_{-}\right)$. In other words, if $g$ were graphed with $y$ on the horizontal axis, and if $Eg$ were graphed with $x$ on the horizontal axis, then, a reflection of $g$ across the vertical axis causes $Eg$ to likewise reflect across the vertical axis.
\end{lem}
\begin{proof}
\begin{align*}
\left(E\left(g_{-}\right)\right)\left(x\right) & = \int \limits _{-\frac{n}{2}} ^{\frac{n}{2}} {\left(r\left(x - y + \frac{a}{2}\right)+r\left(-x + y + \frac{a}{2}\right)\right)g_{-}\left(y\right) dy} \\
& = \int \limits _{-\frac{n}{2}} ^{\frac{n}{2}} {\left(r\left(x - y + \frac{a}{2}\right)+r\left(-x + y + \frac{a}{2}\right)\right)g\left(-y\right) dy} \\
\end{align*}
Substituting $u = -y$:
\begin{align*}
& \int \limits _{-\frac{n}{2}} ^{\frac{n}{2}} {\left(r\left(x - y + \frac{a}{2}\right)+r\left(-x + y + \frac{a}{2}\right)\right)g\left(-y\right) dy} \\
& = -\int \limits _{\frac{n}{2}} ^{-\frac{n}{2}} {\left(r\left(x + u + \frac{a}{2}\right)+r\left(-x - u + \frac{a}{2}\right)\right)g\left(u\right) du} \\
& = \int \limits _{-\frac{n}{2}} ^{\frac{n}{2}} {\left(r\left(x + u + \frac{a}{2}\right)+r\left(-x - u + \frac{a}{2}\right)\right)g\left(u\right) du} \\
& = \int \limits _{-\frac{n}{2}} ^{\frac{n}{2}} {\left(r\left(-x - u + \frac{a}{2}\right)+r\left(x + u + \frac{a}{2}\right)\right)g\left(u\right) du} \\
& = \int \limits _{-\frac{n}{2}} ^{\frac{n}{2}} {\left(r\left(-x - u + \frac{a}{2}\right)+r\left(-\left(-x\right) + u + \frac{a}{2}\right)\right)g\left(u\right) du} \\
\end{align*}
As the choice of the letter used for the variable of integration does not matter:
\begin{align*}
& \int \limits _{-\frac{n}{2}} ^{\frac{n}{2}} {\left(r\left(-x - u + \frac{a}{2}\right)+r\left(-\left(-x\right) + u + \frac{a}{2}\right)\right)g\left(u\right) du} \\
& = \int \limits _{-\frac{n}{2}} ^{\frac{n}{2}} {\left(r\left(-x - y + \frac{a}{2}\right)+r\left(-\left(-x\right) + y + \frac{a}{2}\right)\right)g\left(y\right) du} \\
\end{align*}
Thus:
\[
\left(E\left(g_{-}\right)\right)\left(x\right) = \int \limits _{-\frac{n}{2}} ^{\frac{n}{2}} {\left(r\left(-x - y + \frac{a}{2}\right)+r\left(-\left(-x\right) + y + \frac{a}{2}\right)\right)g\left(y\right) du} \\
\]
Also:
\begin{align*}
\left(\left(Eg\right)_{-}\right)\left(x\right) & = \left(Eg\right)\left(-x\right) \\
& = \int \limits _{-\frac{n}{2}} ^{\frac{n}{2}} {\left(r\left(-x - y + \frac{a}{2}\right)+r\left(-\left(-x\right) + y + \frac{a}{2}\right)\right)g\left(y\right) du} \\
\end{align*}
This implies that $\left(Eg\right)_{-} = E\left(g_{-}\right)$, as these two functions of $x$ agree on every value of $x$.
\end{proof}
\begin{lem} [Effect on $E$ on even and odd parts of $g$]
\label{evenoddlemma}
For all functions $h$, let $h_e = \frac{h+h_{-}}{2}$ and $h_o\left(z\right) = \frac{h-h_{-}}{2}$, with $h_{-}\left(z\right) = h\left(-z\right)$, as defined in Lemma \ref{neglemma}. (In other words, $h_e$ is the "even part" of $h$, and $h_o$ is the "odd part" of $h$.) Then, $\left(Eg\right)_e = E\left(g_e\right)$ and $\left(Eg\right)_o = E\left(g_o\right)$.
\end{lem}
\begin{proof}
\begin{align*}
\left(Eg\right)_e &= \frac{Eg+\left(Eg\right)_{-}}{2} \\
&= \frac{Eg+E\left(g_{-}\right)}{2} \\
&= \frac{E\left(g+g_{-}\right)}{2} \\
&= E\left(\frac{g+g_{-}}{2}\right) \\
&= E\left(g_e\right)
\end{align*}
and, similarly:
\begin{align*}
\left(Eg\right)_o &= \frac{Eg-\left(Eg\right)_{-}}{2} \\
&= \frac{Eg-E\left(g_{-}\right)}{2} \\
&= \frac{E\left(g-g_{-}\right)}{2} \\
&= E\left(\frac{g-g_{-}}{2}\right) \\
&= E\left(g_o\right)
\end{align*}
\end{proof}
\begin{lem}
\label{orthobasislemma}
The orthonormal basis of Player 1's strategies $\left \lbrace f_0\left(x\right), f_1\left(x\right), f_2\left(x\right), \cdots \right \rbrace$ contains only even functions and odd functions. Likewise, the orthonormal basis of Player 1's strategies $\left \lbrace g_0\left(y\right), g_1\left(y\right), g_2\left(y\right), \cdots \right \rbrace$ also contains only even functions and odd functions.
\end{lem}
\begin{proof}
Let $w$ be $x$ for Player 1 or $y$ for Player 2, let $\nu$ be $\frac{n+a}{2}$ for Player 1 or $\frac{n}{2}$ for Player 2, and for all nonnegative integers $s$, let $h_s$ be $f_s$ for Player 1 or $g_s$ for Player 2. In either case, to prove this lemma, it suffices to show that the orthonormal basis $\left \lbrace h_0\left(w\right), h_1\left(w\right), h_2\left(w\right), \cdots \right \rbrace$ contains only even functions and odd functions. For either player, this is the orthonormal (using the dot product $h_a \cdot h_b = \int \limits _{-\nu} ^{\nu} {h_a\left(x\right) h_b\left(x\right) dx}$) basis generated from the basis $\left \lbrace 1, w, w^2, \cdots, w^n \right \rbrace$ by the Gram-Schmidt process. 

Suppose that this is not the case; that is, suppose that, for some nonnegative integer $s$, $h_s$ is neither even nor odd. Of all the $s$, where $h_s$ is neither even nor odd, choose the lowest. $s$ is not $0$, because $h_0 = \frac{1}{\sqrt{1 \cdot 1}}$ is an even function. For $s > 0$, 
\[
h_s = \frac{w^s - \sum \limits _{j = 0} ^{s-1}{\left(\left(w^s \cdot h_j\right)h_j\right)}}{\sqrt{\left({w^s - \sum \limits _{j = 0} ^{s-1}{\left(\left(w^s \cdot h_j\right)h_j\right)}}\right) \cdot \left({w^s - \sum \limits _{j = 0} ^{s-1}{\left(\left(w^s \cdot h_j\right)h_j\right)}}\right)}}
\]
(obtained by reducing $w^s$ to its component perpendicular to all of $h_0, h_1, h_2, \cdots, h_{s-1}$, and by normalizing the result by dividing by the magnitude). This is defined, because, for every polynomial $h$, $h \cdot h =  \int \limits _{-\nu} ^{\nu} {\left(h_j\left(w\right)\right)^2 dw} \geq 0$, and is only zero when $h$ is the zero polynomial; ${w^s - \sum \limits _{j = 0} ^{s-1}{\left(\left(w^s \cdot h_j\right)h_j\right)}}$ is not the zero polynomial, or else $w^s$ is a linear combination of $h_0, h_1, \cdots h_{s-1}$, or, equivalently, a linear combination of $1, w, w^2, \cdots, w^{s-1}$, but this is not the case. 

Now, each $h_j$ is either an even function or an odd function, as is $w^s$. If exactly one of $w^s$ and $h_j$ is an even function (the other one being odd), then $w^s \cdot h_j = \int \limits _{-\nu} ^{\nu} {w^s h_j\left(w\right) dw} = 0$, as this is the integral of an odd function over a symmetric interval. Thus, for the sum $\sum \limits _{j = 0} ^{s-1}{\left(\left(w^s \cdot h_j\right)h_j\right)}$, if $w_s$ is even, then every term in this sum is even\footnote{Each $h_j$ is either even or odd; if $h_j$ is even, then $\left(w^s \cdot h_j\right)h_j$ is even, being a constant multiple of $h_j$; if $h_j$ is odd, then $\left(w^s \cdot h_j\right)h_j = 0 h_j = 0$ is also an even function.}, and, similarly, if $w_s$ is odd, then every term in this sum is odd. It follows that $w^s - \sum \limits _{j = 0} ^{s-1}{\left(\left(w^s \cdot h_j\right)h_j\right)}$ itself is either an even function or an odd function, but $h_s$ is $w^s - \sum \limits _{j = 0} ^{s-1}{\left(\left(w^s \cdot h_j\right)h_j\right)}$ divided by a (nonzero) constant, so $h_s$ is either even or odd, which contradicts the assumption.

From the contrary assumption leading to a contradiction, the orthonormal basis $\left \lbrace h_0\left(w\right), h_1\left(w\right), h_2\left(w\right), \cdots \right \rbrace$ contains only even functions and odd functions. It thus follows that the orthonormal basis of Player 1's strategies $\left \lbrace f_0\left(x\right), f_1\left(x\right), f_2\left(x\right), \cdots \right \rbrace$ contains only even functions and odd functions, and that the orthonormal basis of Player 1's strategies $\left \lbrace g_0\left(y\right), g_1\left(y\right), g_2\left(y\right), \cdots \right \rbrace$ also contains only even functions and odd functions.
\end{proof}
\begin{lem}
\label{oddstrategyzerolemma}
If $f$ and $g$ are functions, then $f_o \cdot \left(Eg_e\right) = 0$, and similarly, $f_e \cdot \left(Eg_o\right) = 0$ (where the subscripts $e$ and $o$ are as defined in Lemma \ref{evenoddlemma}).
\end{lem}
\begin{proof}
\begin{align*}
f_o \cdot \left(Eg_e\right) &= f_o \cdot \left(\left(Eg\right)_e\right) \text{ (by Lemma \ref{evenoddlemma})} \\
&= \int \limits _{-\frac{n+a}{2}} ^{\frac{n+a}{2}} {f_o\left(x\right) \left(Eg\right)_e\left(x\right) dx}
\end{align*}
This is an integral of an odd function (as $f_o\left(x\right) \left(Eg\right)_e$, being the product of an odd function and an even function, is odd) over a symmetric interval, and, therefore, it is zero.
Similarly:
\begin{align*}
f_e \cdot \left(Eg_o\right) &= f_e \cdot \left(\left(Eg\right)_o\right) \text{ (by Lemma \ref{evenoddlemma})} \\
&= \int \limits _{-\frac{n+a}{2}} ^{\frac{n+a}{2}} {f_e\left(x\right) \left(Eg\right)_o\left(x\right) dx}
\end{align*}
This is an integral of an odd function (as $f_e\left(x\right) \left(Eg\right)_o$, being the product of an even function and an odd function, is odd) over a symmetric interval, and, therefore, it is zero.
\end{proof}
\begin{lem} [Even-function equilibrium]
\label{evenequillemma}
If Player 1 playing strategy $f$ and Player 2 playing strategy $g$ is a Nash equilibrium, then so is Player 1 playing strategy $f_e$ and Player 2 playing strategy $g_e$.
\end{lem}
\begin{proof}
Let $f$ and $g$ be strategies that can be played ($f$ by Player 1 and $g$ by Player 2); that is, let $\int \limits _{\frac{n}{2}} ^{\frac{n}{2}} {f\left(x\right) dx} = 1$ and $\int \limits _{\frac{n}{2}} ^{\frac{n}{2}} {g\left(y\right) dx} = 1$, and let $f$ and $g$ be nonnegative on their domains: $f$ on $\left[-\frac{n+a}{2}, \frac{n+a}{2}\right]$, and $g$ on $\left[-\frac{n}{2}, \frac{n}{2}\right]$.
First, it is indeed possible for Player 1 to play strategy $f_e$, as:
\begin{align*}
1 &= \int \limits _{\frac{n+a}{2}} ^{\frac{n+a}{2}} {f\left(x\right) dx} \\
&= \int \limits _{\frac{n+a}{2}} ^{\frac{n+a}{2}} {f_e\left(x\right) + f_o\left(x\right) dx} \\
&= \int \limits _{\frac{n+a}{2}} ^{\frac{n+a}{2}} {f_e\left(x\right) dx} + \int \limits _{\frac{n}{2}} ^{\frac{n}{2}} {f_o\left(x\right) dx} \\
&= \int \limits _{\frac{n+a}{2}} ^{\frac{n+a}{2}} {f_e\left(x\right) dx} + 0 \text { as $f_o$ is an odd function} \\
&= \int \limits _{\frac{n+a}{2}} ^{\frac{n+a}{2}} {f_e\left(x\right) dx}
\end{align*}
and as:
\begin{align*}
f_e\left(x\right) &= \frac{f\left(x\right) + f_{-}\left(x\right)}{2} \\
&= \frac{f\left(x\right)}{2} + \frac{f_{-}\left(x\right)}{2} \\
&= \frac{f\left(x\right)}{2} + \frac{f\left(-x\right)}{2} \\
&\geq \frac{0}{2} + \frac{0}{2} \text{ (as } f\left(x\right) \geq 0 \text { for all } x \in \left[-\frac{n+a}{2}, \frac{n+a}{2}\right] \text{)} \\
&= 0
\end{align*}
It is also possible for Player 2 to play strategy $g_e$, for the same reason.
As $f_o \cdot \left(Eg_e\right) = 0$ (by Lemma \ref{oddstrategyzerolemma}), it follows that:
\begin{align*}
f \cdot \left(Eg_e\right) &= \left(f_e + f_o\right) \cdot \left(Eg_e\right) \\
&= \left(f_e \cdot \left(Eg_e\right)\right) + \left(f_o \cdot \left(Eg_e\right)\right) \\
&= f_e \cdot \left(Eg_e\right) + 0 \\
&= f_e \cdot \left(Eg_e\right)
\end{align*}
Similarly, as $f_e \cdot \left(Eg_o\right) = 0$ (by Lemma \ref{oddstrategyzerolemma}), it follows that:
\begin{align*}
f_e \cdot \left(Eg\right) &= f \cdot \left(E\left(g_e + g_o\right)\right) \\
&= f \cdot \left(Eg_e + Eg_o\right) \\
&= \left(f \cdot \left(Eg_e\right)\right) + \left(f \cdot \left(Eg_o\right)\right) \\
&= f_e \cdot \left(Eg_e\right) + 0 \\
&= f_e \cdot \left(Eg_e\right)
\end{align*}
Let $f$ by Player 1 and $g$ by Player 2 now be a Nash equilibrium. Thus, by definition of the Nash equilibrium:
\begin{align*}
f \cdot \left(Eg\right) & \geq f_e \cdot \left(Eg\right) \\
&= f_e \cdot \left(Eg_e\right) \\
&= f \cdot \left(Eg_e\right) \\
& \geq f \cdot \left(Eg\right)
\end{align*}
which makes all these quantities equal: $f \cdot \left(Eg\right)$,  $f_e \cdot \left(Eg\right)$, $f_e \cdot \left(Eg_e\right)$, and $f \cdot \left(Eg_e\right)$. As shown earlier, it is also the case that $f_e \cdot \left(Eg\right)$, $f_e \cdot \left(Eg_e\right)$, and $f \cdot \left(Eg_e\right)$ (but not necessarily $f \cdot \left(Eg\right)$) would be equal even without the assumption that $f$ by Player 1 and $g$ by Player 2 is a Nash equilibrium.
Now, let Player 1 play strategy $f_e$, and let Player 2 play strategy $g_e$. Then, if Player 1 modifies Player 1's strategy to $\tilde{f}$, the result is $\tilde{f} \cdot \left(Eg_e\right)$. Then:
\begin{align*}
\tilde{f} \cdot \left(Eg_e\right) &= \left(\tilde{f}\right)_e \cdot \left(Eg_e\right) \text{ where $\left(\tilde{f}\right)_e$ is the even part of $\tilde{f}$} \\
&= \left(\tilde{f}\right)_e \cdot \left(Eg\right) \\
& \leq f \cdot \left(Eg\right) \text{ (as $f$ by Player 1 and $g$ by Player 2 is a Nash equilibrium)} \\
& = f_e \cdot \left(Eg_e\right)
\end{align*}
That means that Player 1 cannot benefit from modifying strategy $f_e$ while Player 2 plays strategy $g_e$.
If, instead, Player 2 modifies Player 2's strategy to $\tilde{g}$, the result is $f_e \cdot \left(E\tilde{g}\right)$. Then:
\begin{align*}
f_e \cdot \left(E\tilde{g}\right) &= f_e \cdot \left(E\left(\tilde{g}\right)_e\right) \text{ where $\left(\tilde{g}\right)_e$ is the even part of $\tilde{g}$} \\
&= f \cdot \left(E\left(\tilde{g}\right)_e\right) \\
& \geq f \cdot \left(Eg\right) \text{ (as $f$ by Player 1 and $g$ by Player 2 is a Nash equilibrium)} \\
& = f_e \cdot \left(Eg_e\right)
\end{align*}
That means that Player 2 cannot benefit from modifying strategy $g_e$ while Player 1 plays strategy $f_e$.
Therefore, $f_e$ by Player 1 and $g_e$ by Player 2 is a Nash equilibrium.
\end{proof}

\section{Appendix 2: Calculations}
\begin{calc}
\label{calcR}
If $r\left(z\right) = -z^3$, then:
\begin{align*}
&R\left(x, y\right) \\
& = r\left(x-y + \frac{a}{2}\right) + r\left(-x +y + \frac{a}{2}\right) \\
&= -\left(x-y + \frac{a}{2}\right)^3 - \left(-x +y + \frac{a}{2}\right)^3 \\
&= -\left(x-y\right)^3 - 3\left(x-y\right)^2\left(\frac{a}{2}\right)- 3\left(x-y\right)\left(\frac{a}{2}\right)^2 - \left(\frac{a}{2}\right)^3 \\
& \quad - \left(-x+y\right)^3 - 3\left(-x+y\right)^2\left(\frac{a}{2}\right)- 3\left(-x+y\right)\left(\frac{a}{2}\right)^2 - \left(\frac{a}{2}\right)^3 \\
&= -6\left(x-y\right)^2\left(\frac{a}{2}\right) - 2\left(\frac{a}{2}\right)^3
\end{align*}
\end{calc}

\begin{calc}
\label{calckernel}
If $R\left(x, y\right) = -\left(x - y\right)^2$, then this calculation shows that if $g$ is orthogonal to $1$, $y$, and $y^2$ (which means that $\int \limits _{-\frac{n}{2}} ^{\frac{n}{2}}{g\left(y\right) dy}$, $\int \limits _{-\frac{n}{2}} ^{\frac{n}{2}}{y g\left(y\right) dy}$, and $\int \limits _{-\frac{n}{2}} ^{\frac{n}{2}}{y^2 g\left(y\right) dy}$ are all zero), or if $f$ is orthogonal to $1$, $x$, and $x^2$ (which means that $\int \limits _{-\frac{n+a}{2}} ^{\frac{n+a}{2}}{f\left(x\right) dx}$, $\int \limits _{-\frac{n+a}{2}} ^{\frac{n+a}{2}}{x g\left(x\right) dx}$, and $\int \limits _{-\frac{n+a}{2}} ^{\frac{n+a}{2}}{x^2 g\left(x\right) dx}$ are all zero) then $f \cdot \left(Eg\right) = 0$.

If $g$ is orthogonal to all of $1$, $y$, and $y^2$:
\begin{align*}
Eg &= \int \limits _{-\frac{n}{2}} ^{\frac{n}{2}} {-\left(x - y\right)^2 g\left(y\right)dy} \\
&= \int \limits _{-\frac{n}{2}} ^{\frac{n}{2}} {-\left(x^2 - 2xy + y^2\right)g\left(y\right)dy} \\
&= -x^2 \int \limits _{-\frac{n}{2}} ^{\frac{n}{2}} {-\left(1\right)g\left(y\right)dy} +2x \int \limits _{-\frac{n}{2}} ^{\frac{n}{2}} {-\left(y\right)g\left(y\right)dy} - \int \limits _{-\frac{n}{2}} ^{\frac{n}{2}} {-\left(y^2\right)g\left(y\right)dy} \\
&= -x^2 \left(0\right) +2x \left(0\right) - \left(0\right) \\
&= 0
\end{align*}
Thus, if $g$ is orthogonal to $1$, to $y$, and to $y^2$, then $Eg = 0$. This automatically implies that $f \cdot \left(Eg\right) = \int \limits _{-\frac{n+a}{2}} ^{\frac{n+a}{2}} {f\left(x\right)0 dx} = 0$.

If, instead, $f$ is orthogonal to all of $1$, $x$, and $x^2$:
\begin{align*}
f \cdot \left(Eg\right) &= \int \limits _{-\frac{n+a}{2}} ^{\frac{n+a}{2}} {\int \limits _{-\frac{n}{2}} ^{\frac{n}{2}} {\left(-\left(x-y\right)^2\right)f\left(x\right)g\left(y\right) dy} \,dx} \\
&= \int \limits _{-\frac{n+a}{2}} ^{\frac{n+a}{2}} {\int \limits _{-\frac{n}{2}} ^{\frac{n}{2}} {\left(-x^2 + 2xy - y^2\right)f\left(x\right)g\left(y\right) dy} \,dx} \\
&= \int \limits _{-\frac{n+a}{2}} ^{\frac{n+a}{2}} {\int \limits _{-\frac{n}{2}} ^{\frac{n}{2}} {\left(-x^2 f\left(x\right) + 2x f\left(x\right)y - y^2f\left(x\right)\right)g\left(y\right) dy} \,dx} \\
&= \int \limits _{-\frac{n+a}{2}} ^{\frac{n+a}{2}} {\left(\left(x^2 f\left(x\right)\right)\int \limits _{-\frac{n}{2}} ^{\frac{n}{2}} {-g\left(y\right) dy} + \left(x f\left(x\right)\right)\int \limits _{-\frac{n}{2}} ^{\frac{n}{2}} {2y g\left(y\right) dy} + \left(f\left(x\right)\right)\int \limits _{-\frac{n}{2}} ^{\frac{n}{2}} {-y^2 g\left(y\right) dy}\right) \,dx} \\
&= \int \limits _{-\frac{n+a}{2}} ^{\frac{n+a}{2}} {\left(\left(x^2 f\left(x\right)\right)\int \limits _{-\frac{n}{2}} ^{\frac{n}{2}} {-g\left(y\right) dy}\right) \,dx} +
\int \limits _{-\frac{n+a}{2}} ^{\frac{n+a}{2}} {\left(\left(x f\left(x\right)\right)\int \limits _{-\frac{n}{2}} ^{\frac{n}{2}} {2y g\left(y\right) dy} \right) \,dx} \\
& \quad + \int \limits _{-\frac{n+a}{2}} ^{\frac{n+a}{2}} {\left(\left(f\left(x\right)\right)\int \limits _{-\frac{n}{2}} ^{\frac{n}{2}} {-y^2 g\left(y\right) dy}\right) \,dx} \\
&= \left(\int \limits _{-\frac{n}{2}} ^{\frac{n}{2}} {-g\left(y\right) dy}\right) \int \limits _{-\frac{n+a}{2}} ^{\frac{n+a}{2}} {x^2 f\left(x\right) \,dx} +
\left(\int \limits _{-\frac{n}{2}} ^{\frac{n}{2}} {2y g\left(y\right) dy}\right) \int \limits _{-\frac{n+a}{2}} ^{\frac{n+a}{2}} {x f\left(x\right) \,dx} \\
& \quad +
\left(\int \limits _{-\frac{n}{2}} ^{\frac{n}{2}} {-y^2 g\left(y\right) dy}\right) \int \limits _{-\frac{n+a}{2}} ^{\frac{n+a}{2}} {f\left(x\right) \,dx}
\\
& = 0
\end{align*}
The last step is valid, because the $x$-integrals are all zero, because of the orthogonality of $f\left(x\right)$ to $1$, $x$, and $x^2$. Thus, if $f$ is orthogonal to $1$, to $x$, and to $x^2$, then $f \cdot \left(Eg\right) = \int \limits _{-\frac{n+a}{2}} ^{\frac{n+a}{2}} {f\left(x\right)0 dx} = 0$.

\end{calc}

\begin{calc}
\label{calcortho}
From the polynomials $1, x, x^2, \cdots$, an orthonormal basis can be constructed by using the Gram-Schmidt process. This process is the same for both players, but the results are different, because the dot products are different for both players. The letter $\nu$ should be interpreted as $\frac{n+a}{2}$ for Player 1 and $\frac{n}{2}$ for Player 2. Thus, for both players, $h_1 \cdot h_2 = \int \limits _{-\nu} ^{\nu} {h_1\left(x\right) h_2\left(x\right) dx}$.

To ease computation:
\begin{align*}
x^A \cdot x^B &= \int \limits _{-\nu} ^{\nu} {x^A x^B dx} \\
&= \int \limits _{-\nu} ^{\nu} {x^{A+B} dx} \\
&= \frac{{\nu}^{A+B+1}}{A+B+1} - \frac{{-\nu}^{A+B+1}}{A+B+1} \\
&= \frac{{\nu}^{A+B+1}}{A+B+1} \left(1 - \left(-1\right)^{A+B+1}\right) \\
&= \frac{{\nu}^{A+B+1}}{A+B+1} \left(1 - \left(
\begin{cases}
1, A+B+1 \text{ is even} \\
-1, A+B+1 \text{ is odd} \\
\end{cases}
\right)\right) \\
&= \frac{{\nu}^{A+B+1}}{A+B+1} \left(
\begin{cases}
0, A+B+1 \text{ is even} \\
2, A+B+1 \text{ is odd} \\
\end{cases}
\right)
\end{align*}
Thus, if $A + B$ is even (so $A + B + 1$ is odd), $x^A \cdot x^B = \left(\frac{2}{A+B+1}\right){\nu}^{A+B+1}$, while, if $A + B$ is odd (so $A + B + 1$ is even), $x^A \cdot x^B = 0$.

The first (non-normalized) vector of the basis is $1$. The second vector of the basis is $x - \left(\frac{x \cdot 1}{1 \cdot 1}\right)1$, which is the component of $x$ orthogonal to $1$. As $x \cdot 1 = x \cdot x^0 = 0$, this is just $x$.

The third vector is $x^2 - \left(\frac{x^2 \cdot 1}{1 \cdot 1}\right)1 - \left(\frac{x^2 \cdot x}{x \cdot x}\right)x$, which is orthogonal to both $1$ and $x$.
\begin{align*}
x^2 - \left(\frac{x^2 \cdot 1}{1 \cdot 1}\right)1 - \left(\frac{x^2 \cdot x}{x \cdot x}\right)x &= x^2 - \left(\frac{\left(\frac{2}{2+0+1}\right){\nu}^{2+0+1}}{\left(\frac{2}{0+0+1}\right){\nu}^{0+0+1}}\right)1 - \left(\frac{0}{\left(\frac{2}{1+1+1}\right){\nu}^{1+1+1}}\right)x \\
&= x^2 - \left(\frac{\left(\frac{2}{3}\right){\nu}^{3}}{2\nu}\right)1 - \left(\frac{0}{\left(\frac{2}{3}\right){\nu}^{3}}\right)x \\
&= x^2 - \left(\frac{1}{3}\right){\nu}^2 \\
\end{align*}
Thus, the third vector is $x^2 - \left(\frac{1}{3}\right){\nu}^2$. The fourth vector is: (again, to be orthogonal to $1$, $x$, and $x^2 - \left(\frac{1}{3}\right){\nu}^2$)
\begin{align*}
& x^3 - \left(\frac{x^3 \cdot 1}{1 \cdot 1}\right)1 - \left(\frac{x^3 \cdot x}{x \cdot x}\right)x - \left(\frac{x^3 \cdot \left(x^2 - \left(\frac{1}{3}\right){\nu}^2\right)}{\left(x^2 - \left(\frac{1}{3}\right){\nu}^2\right) \cdot \left(x^2 - \left(\frac{1}{3}\right){\nu}^2\right)}\right)\left(x^2 - \left(\frac{1}{3}\right){\nu}^2\right) \\
&= x^3 - \left(\frac{0}{\left(\frac{2}{1}\right){\nu}^1}\right)1 - \left(\frac{\left(\frac{2}{5}\right){\nu}^5}{\left(\frac{2}{3}\right){\nu}^3}\right)x - \left(\frac{x^3 \cdot x^2 - x^3 \cdot \left(\left(\frac{1}{3}\right){\nu}^2\right)}{\left(x^2 - \left(\frac{1}{3}\right){\nu}^2\right) \cdot \left(x^2 - \left(\frac{1}{3}\right){\nu}^2\right)}\right)\left(x^2 - \left(\frac{1}{3}\right){\nu}^2\right) \\
&= x^3 - \left(\left(\frac{3}{5}\right){\nu}^2\right)x - \left(\frac{0 - 0\left(\left(\frac{1}{3}\right){\nu}^2\right)}{\left(x^2 - \left(\frac{1}{3}\right){\nu}^2\right) \cdot \left(x^2 - \left(\frac{1}{3}\right){\nu}^2\right)}\right)\left(x^2 - \left(\frac{1}{3}\right){\nu}^2\right) \\
&= x^3 - \left(\left(\frac{3}{5}\right){\nu}^2\right)x \\
\end{align*}
(since the dot product of a nonzero with itself is the integral of the square of a nonzero polynomial, and hence is itself nonzero). Thus, the fourth vector is $x^3 - \left(\left(\frac{3}{5}\right){\nu}^2\right)x$. The fifth vector is (to be orthogonal to the first four):
\begin{align*}
& x^4 - \left(\frac{x^4 \cdot 1}{1 \cdot 1}\right)1 - \left(\frac{x^4 \cdot x}{x \cdot x}\right)x - \left(\frac{x^4 \cdot \left(x^2 - \left(\frac{1}{3}\right){\nu}^2\right)}{\left(x^2 - \left(\frac{1}{3}\right){\nu}^2\right) \cdot \left(x^2 - \left(\frac{1}{3}\right){\nu}^2\right)}\right)\left(x^2 - \left(\frac{1}{3}\right){\nu}^2\right) \\
& \quad - \left(\frac{x^4 \cdot \left(x^3 - \left(\frac{3}{5}\right){\nu}^2 x\right)}{\left(x^3 - \left(\frac{3}{5}\right){\nu}^2 x\right) \cdot \left(x^3 - \left(\frac{3}{5}\right){\nu}^2 x\right)}\right)\left(x^3 - \left(\frac{3}{5}\right){\nu}^2 x\right) \\
& = x^4 - \left(\frac{\left(\frac{2}{5}\right){\nu}^5}{\left(\frac{2}{1}\right){\nu}^1}\right)1 - \left(\frac{0}{x \cdot x}\right)x - \left(\frac{x^4 \cdot x^2 - \left(\frac{1}{3}\right){\nu}^2\left(x^4 \cdot 1 \right)}{\left(x^2 - \left(\frac{1}{3}\right){\nu}^2\right) \cdot \left(x^2 - \left(\frac{1}{3}\right){\nu}^2\right)}\right)\left(x^2 - \left(\frac{1}{3}\right){\nu}^2\right) \\
& \quad - \left(\frac{x^4 \cdot x^3 - \left(\frac{3}{5}\right){\nu}^2 \left(x^4 \cdot x\right)}{\left(x^3 - \left(\frac{3}{5}\right){\nu}^2 x\right) \cdot \left(x^3 - \left(\frac{3}{5}\right){\nu}^2 x\right)}\right)\left(x^3 - \left(\frac{3}{5}\right){\nu}^2 x\right) \\
& = x^4 - \left(\frac{1}{5}\right){\nu}^4 - \left(\frac{\left(\frac{2}{7}\right){\nu}^7 - \left(\left(\frac{1}{3}\right){\nu}^2\right)\left(\left(\frac{2}{5}\right)\nu^{5}\right)}{x^2 \cdot x^2 - 2\left(\left(\frac{1}{3}\right){\nu}^2\right)\left(x^2 \cdot 1\right) + \left(\left(\frac{1}{3}\right){\nu}^2\right)^2\left(1 \cdot 1\right)}\right)\left(x^2 - \left(\frac{1}{3}\right){\nu}^2\right) \\
& \quad - \left(\frac{0 - \left(\frac{3}{5}\right){\nu}^2 \left(0\right)}{\left(x^3 - \left(\frac{3}{5}\right){\nu}^2 x\right) \cdot \left(x^3 - \left(\frac{3}{5}\right){\nu}^2 x\right)}\right)\left(x^3 - \left(\frac{3}{5}\right){\nu}^2 x\right) \\
& = x^4 - \left(\frac{1}{5}\right){\nu}^4 - \left(\frac{\left(\frac{2}{7}\right){\nu}^7 - \left(\frac{2}{15}\right){\nu}^7}{\left(\frac{2}{5}\right){\nu}^5 - 2\left(\left(\frac{1}{3}\right){\nu}^2\right)\left(\left(\frac{2}{3}\right){\nu}^3\right) + \left(\left(\frac{1}{3}\right){\nu}^2\right)^2\left(\left(\frac{2}{1}\right){\nu}^1\right)}\right)\left(x^2 - \left(\frac{1}{3}\right){\nu}^2\right)
\end{align*}
This simplifies to:
\begin{align*}
x^4 - \left(\frac{1}{5}\right){\nu}^4 - \left(\frac{\left(\frac{16}{105}\right){\nu}^7}{\left(\frac{8}{45}\right){\nu}^5}\right)\left(x^2 - \left(\frac{1}{3}\right){\nu}^2\right) &= x^4 - \left(\frac{1}{5}\right){\nu}^4 - \left(\left(\frac{6}{7}\right){\nu}^2\right)\left(x^2 - \left(\frac{1}{3}\right){\nu}^2\right) \\
&= x^4 - \left(\frac{1}{5}\right){\nu}^4 - \left(\frac{6}{7}\right){\nu}^2 x^2 + \left(\frac{2}{7}\right){\nu}^4 \\
&= x^4 - \left(\frac{6}{7}\right){\nu}^2 x^2 + \left(\frac{3}{35}\right){\nu}^4 \\
\end{align*}
Thus, the first five non-normalized vectors are $1$, $x$, $x^2 - \left(\frac{1}{3}\right){\nu}^2$, $x^3 - \left(\frac{3}{5}\right){\nu}^2 x$, and $x^4 - \left(\frac{6}{7}\right){\nu}^2 x^2 + \left(\frac{3}{35}\right){\nu}^4$. Normalizing them requires dividing each of them by its magnitude. Thus, the first vector is:
\begin{align*}
\frac{1}{\sqrt{1 \cdot 1}} &= \frac{1}{\sqrt{\left(\frac{2}{1}\right){\nu}^1}} \\
&= \left(\sqrt{\frac{1}{2}}\right){\nu}^{-\frac{1}{2}}
\end{align*}
The second vector is:
\begin{align*}
\frac{x}{\sqrt{x \cdot x}} &= \frac{x}{\sqrt{\left(\frac{2}{3}\right){\nu}^3}} \\
&= \left(\sqrt{\frac{3}{2}}\right){\nu}^{-\frac{3}{2}} x \\
\end{align*}
The third vector is:
\begin{align*}
\frac{x^2 - \left(\frac{1}{3}\right){\nu}^2}{\sqrt{\left(x^2 - \left(\frac{1}{3}\right){\nu}^2\right) \cdot \left(x^2 - \left(\frac{1}{3}\right){\nu}^2\right)}} &= \frac{x^2 - \left(\frac{1}{3}\right){\nu}^2}{\sqrt{x^2 \cdot x^2 - 2\left(\left(\frac{1}{3}\right){\nu}^2\right)\left(x^2 \cdot 1\right) + \left(\left(\frac{1}{3}\right){\nu}^2\right)^2\left(1 \cdot 1\right)}} \\
&= \frac{x^2 - \left(\frac{1}{3}\right){\nu}^2}{\sqrt{\left(\frac{2}{5}\right){\nu}^5 - 2\left(\left(\frac{1}{3}\right){\nu}^2\right)\left(\left(\frac{2}{3}\right){\nu}^3\right) + \left(\left(\frac{1}{3}\right){\nu}^2\right)^2\left(\left(\frac{2}{1}\right){\nu}^1\right)}} \\
&= \frac{x^2 - \left(\frac{1}{3}\right){\nu}^2}{\sqrt{\left(\frac{2}{5}\right){\nu}^5 - \left(\frac{4}{9}\right){\nu}^5 + \left(\frac{2}{9}\right){\nu}^5}} \\
&= \frac{x^2 - \left(\frac{1}{3}\right){\nu}^2}{\sqrt{\left(\frac{8}{45}\right){\nu}^5}} \\
&= \left(\sqrt{\frac{45}{8}}\right){\nu}^{-\frac{5}{2}}x^2 - \left(\sqrt{\frac{5}{8}}\right){\nu}^{-\frac{1}{2}} \\
\end{align*}
The fourth vector is:
\begin{align*}
\frac{x^3 - \left(\frac{3}{5}\right){\nu}^2 x}{\sqrt{\left(x^3 - \left(\frac{3}{5}\right){\nu}^2 x\right) \cdot \left(x^3 - \left(\frac{3}{5}\right){\nu}^2 x\right)}} &= \frac{x^3 - \left(\frac{3}{5}\right){\nu}^2 x}{\sqrt{x^3 \cdot x^3 - 2\left(\left(\frac{3}{5}\right){\nu}^2\right)\left(x^3 \cdot x\right) + \left(\left(\frac{3}{5}\right){\nu}^2\right)^2\left(x \cdot x\right)}} \\
&= \frac{x^3 - \left(\frac{3}{5}\right){\nu}^2 x}{\sqrt{\left(\frac{2}{7}\right){\nu}^7 - 2\left(\left(\frac{3}{5}\right){\nu}^2\right)\left(\left(\frac{2}{5}\right){\nu}^5\right) + \left(\left(\frac{3}{5}\right){\nu}^2\right)^2\left(\left(\frac{2}{3}\right){\nu}^3\right)}} \\
&= \frac{x^3 - \left(\frac{3}{5}\right){\nu}^2 x}{\sqrt{\left(\frac{2}{7}\right){\nu}^7 - \left(\frac{12}{25}\right){\nu}^7 + \left(\frac{6}{25}\right){\nu}^7}} \\
&= \frac{x^3 - \left(\frac{3}{5}\right){\nu}^2 x}{\sqrt{\left(\frac{8}{175}\right){\nu}^7}} \\
&= \left(\sqrt{\frac{175}{8}}\right){\nu}^{-\frac{7}{2}}x^3 - \left(\sqrt{\frac{63}{8}}\right){\nu}^{-\frac{3}{2}}x \\
\end{align*}
The fifth vector is:
\begin{align*}
& \frac{x^4 - \left(\frac{6}{7}\right){\nu}^2 x^2 + \left(\frac{3}{35}\right){\nu}^4}{\sqrt{\left(x^4 - \left(\frac{6}{7}\right){\nu}^2 x^2 + \left(\frac{3}{35}\right){\nu}^4\right) \cdot \left(x^4 - \left(\frac{6}{7}\right){\nu}^2 x^2 + \left(\frac{3}{35}\right){\nu}^4\right)}} \\
& = \frac{x^4 - \left(\frac{6}{7}\right){\nu}^2 x^2 + \left(\frac{3}{35}\right){\nu}^4}{\sqrt{
\begin{aligned}
& \qquad x^4 \cdot x^4 - 2\left(\left(\frac{6}{7}\right){\nu}^2\right)\left(x^4 \cdot x^2\right) + 2\left(\left(\frac{3}{35}\right){\nu}^4\right)\left(x^4 \cdot 1\right) \\
& + \left(\left(\frac{6}{7}\right){\nu}^2\right)^2\left(x^2 \cdot x^2\right) - 2\left(\left(\frac{6}{7}\right){\nu}^2\right)\left(\left(\frac{3}{35}\right){\nu}^4\right)\left(x^2 \cdot 1\right) + \left(\left(\frac{3}{35}\right){\nu}^4\right)^2\left(1 \cdot 1\right)
\end{aligned}
}} \\
& \hspace{-0.6 in} = \frac{x^4 - \left(\frac{6}{7}\right){\nu}^2 x^2 + \left(\frac{3}{35}\right){\nu}^4}{\sqrt{
\begin{aligned}
& \left(\frac{2}{9}\right){\nu}^9 - 2\left(\left(\frac{6}{7}\right){\nu}^2\right)\left(\left(\frac{2}{7}\right){\nu}^7\right) + 2\left(\left(\frac{3}{35}\right){\nu}^4\right)\left(\left(\frac{2}{5}\right){\nu}^5\right) + \left(\left(\frac{6}{7}\right){\nu}^2\right)^2\left(\left(\frac{2}{5}\right){\nu}^5\right) \\
& \quad - 2\left(\left(\frac{6}{7}\right){\nu}^2\right)\left(\left(\frac{3}{35}\right){\nu}^4\right)\left(\left(\frac{2}{3}\right){\nu}^3\right) + \left(\left(\frac{3}{35}\right){\nu}^4\right)^2\left(\left(\frac{2}{1}\right){\nu}^1\right)
\end{aligned}
}} \\
\end{align*}
This simplifies to:
\begin{align*}
& \frac{x^4 - \left(\frac{6}{7}\right){\nu}^2 x^2 + \left(\frac{3}{35}\right){\nu}^4}{\sqrt{\left(\frac{2}{9}\right){\nu}^9 - \left(\frac{24}{49}\right){\nu}^9 + \left(\frac{12}{175}\right){\nu}^9+\left(\frac{72}{245}\right){\nu}^9 - \left(\frac{24}{245}\right){\nu}^9+\left(\frac{18}{1225}\right){\nu}^9}} \\
&= \frac{x^4 - \left(\frac{6}{7}\right){\nu}^2 x^2 + \left(\frac{3}{35}\right){\nu}^4}{\sqrt{\left(\frac{128}{11025}\right){\nu}^9 }} \\
&= \left(\sqrt{\frac{11025}{128}}\right){\nu}^{-\frac{9}{2}}x^4 - \left(\sqrt{\frac{2025}{32}}\right){\nu}^{-\frac{5}{2}} x^2 + \left(\sqrt{\frac{81}{128}}\right){\nu}^{-\frac{1}{2}}
\end{align*}
Thus, the first five orthonormal polynomials are:
\begin{align*}
& \left(\sqrt{\frac{1}{2}}\right){\nu}^{-\frac{1}{2}} \\
& \left(\sqrt{\frac{3}{2}}\right){\nu}^{-\frac{3}{2}} x \\
& \left(\sqrt{\frac{45}{8}}\right){\nu}^{-\frac{5}{2}}x^2 - \left(\sqrt{\frac{5}{8}}\right){\nu}^{-\frac{1}{2}} \\
& \left(\sqrt{\frac{175}{8}}\right){\nu}^{-\frac{7}{2}}x^3 - \left(\sqrt{\frac{63}{8}}\right){\nu}^{-\frac{3}{2}}x \\
& \left(\sqrt{\frac{11025}{128}}\right){\nu}^{-\frac{9}{2}}x^4 - \left(\sqrt{\frac{2025}{32}}\right){\nu}^{-\frac{5}{2}} x^2 + \left(\sqrt{\frac{81}{128}}\right){\nu}^{-\frac{1}{2}}
\end{align*}
(where $\nu$ is $\frac{n+a}{2}$ for Player 1 and $\frac{n}{2}$ for Player 2).
\end{calc}

\begin{calc}
\label{calcpurestrategies}
The pure strategy with $x = t$, with $t \in \left[-\nu, \nu\right]$, where $\nu$ is $\frac{n+a}{2}$ for Player 1 and $\frac{n}{2}$ for Player 2, is represented by the Dirac delta ``function'' $\delta\left(x - t\right)$. For any function $h$, the following holds:
\begin{align*}
\delta\left(x - t\right) \cdot h &= \int \limits _{-\nu} ^{\nu} {\delta\left(x - t\right) h\left(x\right) dx} \\
&= \int \limits _{-\nu} ^{\nu} {\delta\left(x - t\right) h\left(t\right) dx}  \\
&= h\left(t\right) \int \limits _{-\nu} ^{\nu} {\delta\left(x - t\right) dx} \\
&= \left(h\left(t\right)\right) \left(1\right) \\
&= h\left(t\right)
\end{align*}
This means that the pure strategies for either player in the new coordinates are
$
\left(
\begin{matrix}
f_1\left(t\right) \\
f_2\left(t\right) \\
f_3\left(t\right)
\end{matrix}
\right)
$ for Player 1 and 
$
\left(
\begin{matrix}
g_1\left(t\right) \\
g_2\left(t\right) \\
g_3\left(t\right)
\end{matrix}
\right)
$ for Player 2, which, equivalently, are 
\[
\left(
\begin{matrix}
\left(\sqrt{\frac{1}{2}}\right){\left(\frac{n+a}{2}\right)}^{-\frac{1}{2}} \\
\left(\sqrt{\frac{3}{2}}\right){\left(\frac{n+a}{2}\right)}^{-\frac{3}{2}} t \\
\left(\sqrt{\frac{45}{8}}\right){\left(\frac{n+a}{2}\right)}^{-\frac{5}{2}} t^2 - \left(\sqrt{\frac{5}{8}}\right){\left(\frac{n+a}{2}\right)}^{-\frac{1}{2}} \\
\end{matrix}
\right)
\]
for Player 1 and
\[
\left(
\begin{matrix}
\left(\sqrt{\frac{1}{2}}\right){\left(\frac{n}{2}\right)}^{-\frac{1}{2}} \\
\left(\sqrt{\frac{3}{2}}\right){\left(\frac{n}{2}\right)}^{-\frac{3}{2}} t \\
\left(\sqrt{\frac{45}{8}}\right){\left(\frac{n}{2}\right)}^{-\frac{5}{2}} t^2 - \left(\sqrt{\frac{5}{8}}\right){\left(\frac{n}{2}\right)}^{-\frac{1}{2}} \\
\end{matrix}
\right)
\]
for Player 2.
\end{calc}
\begin{calc}
\label{symmetrycalc}
The curve
\[
\left(
\begin{matrix}
\left(\sqrt{\frac{1}{2}}\right){\left(\frac{n+a}{2}\right)}^{-\frac{1}{2}} \\
\left(\sqrt{\frac{3}{2}}\right){\left(\frac{n+a}{2}\right)}^{-\frac{3}{2}} t \\
\left(\sqrt{\frac{45}{8}}\right){\left(\frac{n+a}{2}\right)}^{-\frac{5}{2}} t^2 - \left(\sqrt{\frac{5}{8}}\right){\left(\frac{n+a}{2}\right)}^{-\frac{1}{2}} \\
\end{matrix}
\right)
\]
of Player 1's pure strategies is symmetrical over the plane $x_2 = 0$. That is, if $\left(x_1, x_2, x_3\right)$ is on the curve, then so is $\left(x_1, -x_2, x_3\right)$, which can be reached simply by switching the sign of $t$. It follows that the convex hull of this curve is also symmetric over the plane $x_2 = 0$, for the same reason: in any convex combination of pure strategies, switching the sign of $t$ in all the components switches the sign of the third component, but leaves the other two components as they were.

Similarly, the convex hull of the curve of Player 2's pure strategies, which is the convex hull of the curve
\[
\left(
\begin{matrix}
\left(\sqrt{\frac{1}{2}}\right){\left(\frac{n}{2}\right)}^{-\frac{1}{2}} \\
\left(\sqrt{\frac{3}{2}}\right){\left(\frac{n}{2}\right)}^{-\frac{3}{2}} t \\
\left(\sqrt{\frac{45}{8}}\right){\left(\frac{n}{2}\right)}^{-\frac{5}{2}} t^2 - \left(\sqrt{\frac{5}{8}}\right){\left(\frac{n}{2}\right)}^{-\frac{1}{2}} \\
\end{matrix}
\right)
\]
is symmetrical over the plane $y_2 = 0$.

This means that for every choice of $f \cdot f_0$ and $f \cdot f_2$ Player 1 can make, Player 1 can always set $f \cdot f_1$ to be zero without leaving the convex hull, and a similar statement holds for Player 2.

Player 1's expected payoff $f \cdot \left(Eg\right)$, which is
\[
\hspace{-0.8 in}
\left( 
\begin{matrix}
f \cdot f_0 \\
f \cdot f_1 \\
f \cdot f_2 \\
\end{matrix}
\right)
^T
\left(
\begin{matrix}
\left(-\frac{2}{3}\right)\left(\frac{n+a}{2}\right)^{\frac{5}{2}}\left(\frac{n}{2}\right)^{\frac{1}{2}} + \left(-\frac{2}{3}\right)\left(\frac{n+a}{2}\right)^{\frac{5}{2}}\left(\frac{n}{2}\right)^{\frac{1}{2}} & 0 & \left(-\frac{4\sqrt{5}}{15}\right)\left(\frac{n+a}{2}\right)^{\frac{1}{2}}\left(\frac{n}{2}\right)^{\frac{5}{2}} \\
0 & \left(-\frac{4}{3}\right)\left(\frac{n+a}{2}\right)^{\frac{3}{2}}\left(\frac{n}{2}\right)^{\frac{3}{2}} & 0 \\
\left(-\frac{4\sqrt{5}}{15}\right)\left(\frac{n+a}{2}\right)^{\frac{5}{2}}\left(\frac{n}{2}\right)^{\frac{1}{2}} & 0 & 0 \\
\end{matrix}
\right)
\left( 
\begin{matrix}
g \cdot g_0 \\
g \cdot g_1 \\
g \cdot g_2 \\
\end{matrix}
\right)
\]
can be rewritten as
\[
\hspace{-0.8 in}
\left( 
\begin{matrix}
f \cdot f_0 \\
f \cdot f_2 \\
f \cdot f_1 \\
\end{matrix}
\right)
^T
\left(
\begin{matrix}
\left(-\frac{2}{3}\right)\left(\frac{n+a}{2}\right)^{\frac{5}{2}}\left(\frac{n}{2}\right)^{\frac{1}{2}} + \left(-\frac{2}{3}\right)\left(\frac{n+a}{2}\right)^{\frac{5}{2}}\left(\frac{n}{2}\right)^{\frac{1}{2}} & \left(-\frac{4\sqrt{5}}{15}\right)\left(\frac{n+a}{2}\right)^{\frac{1}{2}}\left(\frac{n}{2}\right)^{\frac{5}{2}} & 0 \\
\left(-\frac{4\sqrt{5}}{15}\right)\left(\frac{n+a}{2}\right)^{\frac{5}{2}}\left(\frac{n}{2}\right)^{\frac{1}{2}} & 0 & 0 \\
0 & 0 & \left(-\frac{4}{3}\right)\left(\frac{n+a}{2}\right)^{\frac{3}{2}}\left(\frac{n}{2}\right)^{\frac{3}{2}} \\
\end{matrix}
\right)
\left( 
\begin{matrix}
g \cdot g_0 \\
g \cdot g_2 \\
g \cdot g_1 \\
\end{matrix}
\right)
\]
or as
\begin{align*}
& \left( 
\begin{matrix}
f \cdot f_0 \\
f \cdot f_2 \\
\end{matrix}
\right)
^T
\left(
\begin{matrix}
\left(-\frac{2}{3}\right)\left(\frac{n+a}{2}\right)^{\frac{5}{2}}\left(\frac{n}{2}\right)^{\frac{1}{2}} + \left(-\frac{2}{3}\right)\left(\frac{n+a}{2}\right)^{\frac{5}{2}}\left(\frac{n}{2}\right)^{\frac{1}{2}} & \left(-\frac{4\sqrt{5}}{15}\right)\left(\frac{n+a}{2}\right)^{\frac{1}{2}}\left(\frac{n}{2}\right)^{\frac{5}{2}} \\
\left(-\frac{4\sqrt{5}}{15}\right)\left(\frac{n+a}{2}\right)^{\frac{5}{2}}\left(\frac{n}{2}\right)^{\frac{1}{2}} & 0 \\
\end{matrix}
\right)
\left( 
\begin{matrix}
g \cdot g_0 \\
g \cdot g_2 \\
\end{matrix}
\right) \\
& \quad + \left(-\frac{4}{3}\right)\left(\frac{n+a}{2}\right)^{\frac{3}{2}}\left(\frac{n}{2}\right)^{\frac{3}{2}}\left(f \cdot f_1\right)\left(g \cdot g_1 \right)
\end{align*}
(as can be verified by carrying out the matrix multiplications). Either player can set $\left(-\frac{4}{3}\right)\left(\frac{n+a}{2}\right)^{\frac{3}{2}}\left(\frac{n}{2}\right)^{\frac{3}{2}}\left(f \cdot f_1\right)\left(g \cdot g_1 \right)$ to be zero without affecting the other term, by setting $f \cdot f_1$ or $g \cdot g_1$ to zero. It follows that in any Nash equilibrium, Player 1 setting $f \cdot f_1$ to zero, and Player 2 setting $g \cdot g_1$ to zero, and both players doing this, yield the same expected payoff. Thus, Player 1 setting $f \cdot f_1$ to zero gives no new options to Player 2, as Player 2 could simply have "simulated" those options by setting $g \cdot g_1$ to zero. Thus, the result of both $f \cdot f_1$ and $g \cdot g_1$ being set to zero is another Nash equilibrium.

Thus, Player 1's strategy space can be reduced to the convex hull of
$
\left(
\begin{matrix}
\left(\sqrt{\frac{1}{2}}\right){\left(\frac{n+a}{2}\right)}^{-\frac{1}{2}} \\
0 \\
\left(\sqrt{\frac{45}{8}}\right){\left(\frac{n+a}{2}\right)}^{-\frac{5}{2}} t^2 - \left(\sqrt{\frac{5}{8}}\right){\left(\frac{n+a}{2}\right)}^{-\frac{1}{2}} \\
\end{matrix}
\right)
$, 
and Player 2's strategy space can be reduced to the convex hull of
$
\left(
\begin{matrix}
\left(\sqrt{\frac{1}{2}}\right){\left(\frac{n}{2}\right)}^{-\frac{1}{2}} \\
0 \\
\left(\sqrt{\frac{45}{8}}\right){\left(\frac{n}{2}\right)}^{-\frac{5}{2}} t^2 - \left(\sqrt{\frac{5}{8}}\right){\left(\frac{n}{2}\right)}^{-\frac{1}{2}} \\
\end{matrix}
\right)
$, 
with $t$ going from $-\frac{n+a}{2}$ to $\frac{n+a}{2}$ for Player 1 and from $-\frac{n}{2}$ to $\frac{n}{2}$ for Player 2, the interval in which the pure strategies fall. This interval can be cut in half to $\left[0, \frac{n+a}{2}\right]$ for Player 1 and to $\left[0, \frac{n}{2}\right]$ for Player 2, because the negative values of $t$ contribute no points to either curve that were not already contributed by a positive value of $t$, and thus they contribute no new points to the convex hull of that curve. This makes $t^2$ a bijective function, so $T_x = \left(\frac{n+a}{2}\right)^{-2} t^2$ for Player 1 and $T_y = \left(\frac{n}{2}\right)^{-2} t^2$ can serve as the parameters for these curves. Both $T_x$ and $T_y$ are in $\left[0, 1\right]$. In terms of $T_x$, Player 1's strategy space can be reduced to the convex hull of
$
\left(
\begin{matrix}
\left(\sqrt{\frac{1}{2}}\right){\left(\frac{n+a}{2}\right)}^{-\frac{1}{2}} \\
0 \\
\left(\sqrt{\frac{45}{8}}\right){\left(\frac{n+a}{2}\right)}^{-\frac{1}{2}} T_x - \left(\sqrt{\frac{5}{8}}\right){\left(\frac{n+a}{2}\right)}^{-\frac{1}{2}} \\
\end{matrix}
\right)
$, 
and Player 2's strategy space can be reduced to the convex hull of
$
\left(
\begin{matrix}
\left(\sqrt{\frac{1}{2}}\right){\left(\frac{n}{2}\right)}^{-\frac{1}{2}} \\
0 \\
\left(\sqrt{\frac{45}{8}}\right){\left(\frac{n}{2}\right)}^{-\frac{1}{2}} T_y - \left(\sqrt{\frac{5}{8}}\right){\left(\frac{n}{2}\right)}^{-\frac{1}{2}} \\
\end{matrix}
\right)
$. These are line segments (although they would not be line segments if $r$ were of higher degree than $3$), so they are their own convex hulls. 

After factoring, these segments become 
$
\left(\sqrt{\frac{1}{2}}\right){\left(\frac{n+a}{2}\right)}^{-\frac{1}{2}}
\left(
\begin{matrix}
1 \\
0 \\
\left(\sqrt{\frac{45}{4}}\right) T_x - \sqrt{\frac{5}{4}} \\
\end{matrix}
\right)
$
 for Player 1 and
$
\left(\sqrt{\frac{1}{2}}\right){\left(\frac{n}{2}\right)}^{-\frac{1}{2}}
\left(
\begin{matrix}
1 \\
0 \\
\left(\sqrt{\frac{45}{4}}\right) T_y - \sqrt{\frac{5}{4}} \\
\end{matrix}
\right)
$
 for Player 2.

As these reduced spaces are in the new coordinates, Player 1's payoff is
\begin{align*}
& \left(\sqrt{\frac{1}{2}}\right){\left(\frac{n+a}{2}\right)}^{-\frac{1}{2}}
\left(
\begin{matrix}
1 \\
0 \\
\left(\sqrt{\frac{45}{4}}\right) T_x - \sqrt{\frac{5}{4}} \\
\end{matrix}
\right)
^T \\
& \quad *
\left(
\begin{matrix}
\left(-\frac{2}{3}\right)\left(\frac{n+a}{2}\right)^{\frac{5}{2}}\left(\frac{n}{2}\right)^{\frac{1}{2}} + \left(-\frac{2}{3}\right)\left(\frac{n+a}{2}\right)^{\frac{5}{2}}\left(\frac{n}{2}\right)^{\frac{1}{2}} & 0 & \left(-\frac{4\sqrt{5}}{15}\right)\left(\frac{n+a}{2}\right)^{\frac{1}{2}}\left(\frac{n}{2}\right)^{\frac{5}{2}} \\
0 & \left(-\frac{4}{3}\right)\left(\frac{n+a}{2}\right)^{\frac{3}{2}}\left(\frac{n}{2}\right)^{\frac{3}{2}} & 0 \\
\left(-\frac{4\sqrt{5}}{15}\right)\left(\frac{n+a}{2}\right)^{\frac{5}{2}}\left(\frac{n}{2}\right)^{\frac{1}{2}} & 0 & 0 \\
\end{matrix}
\right) \\
& \quad * \left(\left(\sqrt{\frac{1}{2}}\right){\left(\frac{n}{2}\right)}^{-\frac{1}{2}}
\left(
\begin{matrix}
1 \\
0 \\
\left(\sqrt{\frac{45}{4}}\right) T_y - \sqrt{\frac{5}{4}} \\
\end{matrix}
\right)
\right) \\
& = \left(\frac{1}{2}\right){\left(\frac{n+a}{2}\right)}^{-\frac{1}{2}}{\left(\frac{n}{2}\right)}^{-\frac{1}{2}}
\left(
\begin{matrix}
1 \\
0 \\
\left(\sqrt{\frac{45}{4}}\right) T_x - \sqrt{\frac{5}{4}} \\
\end{matrix}
\right)
^T \\
& \quad *
\left(
\begin{matrix}
\left(-\frac{2}{3}\right)\left(\frac{n+a}{2}\right)^{\frac{5}{2}}\left(\frac{n}{2}\right)^{\frac{1}{2}} + \left(-\frac{2}{3}\right)\left(\frac{n+a}{2}\right)^{\frac{5}{2}}\left(\frac{n}{2}\right)^{\frac{1}{2}} & 0 & \left(-\frac{4\sqrt{5}}{15}\right)\left(\frac{n+a}{2}\right)^{\frac{1}{2}}\left(\frac{n}{2}\right)^{\frac{5}{2}} \\
0 & \left(-\frac{4}{3}\right)\left(\frac{n+a}{2}\right)^{\frac{3}{2}}\left(\frac{n}{2}\right)^{\frac{3}{2}} & 0 \\
\left(-\frac{4\sqrt{5}}{15}\right)\left(\frac{n+a}{2}\right)^{\frac{5}{2}}\left(\frac{n}{2}\right)^{\frac{1}{2}} & 0 & 0 \\
\end{matrix}
\right) \\
& \quad * 
\left(
\begin{matrix}
1 \\
0 \\
\left(\sqrt{\frac{45}{4}}\right) T_y - \sqrt{\frac{5}{4}} \\
\end{matrix}
\right) \\
& = \left(\frac{1}{2}\right){\left(\frac{n+a}{2}\right)}^{-\frac{1}{2}}{\left(\frac{n}{2}\right)}^{-\frac{1}{2}}\left(
\begin{aligned}
& \left(\left(-\frac{2}{3}\right)\left(\frac{n+a}{2}\right)^{\frac{5}{2}}\left(\frac{n}{2}\right)^{\frac{1}{2}} + \left(-\frac{2}{3}\right)\left(\frac{n+a}{2}\right)^{\frac{5}{2}}\left(\frac{n}{2}\right)^{\frac{1}{2}}\right)\left(1\right)\left(1\right) \\
& \quad + \left(-\frac{4\sqrt{5}}{15}\right)\left(\frac{n+a}{2}\right)^{\frac{5}{2}}\left(\frac{n}{2}\right)^{\frac{1}{2}}\left(\left(\sqrt{\frac{45}{4}}\right) T_x - \sqrt{\frac{5}{4}}\right)\left(1\right) \\
& \quad + \left(-\frac{4\sqrt{5}}{15}\right)\left(\frac{n+a}{2}\right)^{\frac{5}{2}}\left(\frac{n}{2}\right)^{\frac{1}{2}}\left(1\right)\left(\left(\sqrt{\frac{45}{4}}\right) T_y - \sqrt{\frac{5}{4}}\right)
\end{aligned}
\right) \\
\end{align*}
(by multiplying out the matrix product and ignoring any zero entries).

This is a decreasing function in both $T_x$ and $T_y$, so Player 1 should choose $T = 0$ and Player 2 should choose $T = 1$. Player 1's choice corresponds to $t = 0$, and hence, to the pure strategy 0. Player 2's choice corresponds to $t = \frac{n}{2}$, and hence, to the equal mixture of pure stategies $\frac{n}{2}$ and $-\frac{n}{2}$.
\end{calc}
\bibliographystyle{plain}

\end{document}